\documentclass[reqno, a4paper, 10pt]{amsart}

\usepackage{amsmath,bbm}%
\usepackage{amsfonts}%
\usepackage{amssymb}%
\usepackage{graphicx, color}
\usepackage{colonequals,nicefrac}
\newtheorem{theorem}{Theorem}\numberwithin{theorem}{section}
\theoremstyle{plain} 

\newtheorem{assumption}{Assumption}

\newtheorem{corollary}{Corollary}\numberwithin{corollary}{section}

\newtheorem{definition}{Definition}\numberwithin{definition}{section}

\newtheorem{lemma}{Lemma}\numberwithin{lemma}{section}

\newtheorem{proposition}{Proposition}\numberwithin{proposition}{section}
\newtheorem{remark}{Remark}

\newtheorem*{pf_appendix_1}{Proof of Proposition 3.1}
\newtheorem*{pf_appendix_2}{Proof of Theorem 3.1}

\numberwithin{equation}{section}
\begin{document}
\title[Long-Term Yield in an Affine HJM Framework on $S_{d}^{+}$]{Long-Term\,\,Yield\,\,in\,\,an\,\,Affine\,\,HJM\,\,Framework\,\,on\,\,$S_{d}^{+}$}
\author[F. Biagini]{Francesca Biagini}
\address[F. Biagini]{Department of Mathematics, LMU University\newline \indent Theresienstrasse 39, D-80333 Munich, Germany}
\email{francesca.biagini@math.lmu.de}%
\urladdr{http://www.mathematik.uni-muenchen.de/personen/professoren/biagini/index.html}
\author[A. Gnoatto]{Alessandro Gnoatto}
\address[A. Gnoatto]{Department of Mathematics, LMU University\newline \indent Theresienstrasse 39, D-80333 Munich, Germany}
\email{alessandro.gnoatto@math.lmu.de}%
\urladdr{http://www.mathematik.uni-muenchen.de/personen/mitarbeiter/gnoatto/index.html}
\author[M. H\"{a}rtel]{Maximilian H\"{a}rtel}
\address[M. H\"{a}rtel]{Department of Mathematics, LMU University\newline \indent Theresienstrasse 39, D-80333 Munich, Germany}
\email{maximilian.haertel@math.lmu.de}%
\urladdr{http://www.mathematik.uni-muenchen.de/personen/mitarbeiter/haertel/index.html}
 \thanks{The research leading to these results has received funding from the European Research Council under
the European Community's Seventh Framework Programme (FP7/2007-2013) / ERC grant agreement no [228087].}
\date{\today}
\subjclass{} %
\keywords{HJM, Affine Process, Long-Term Yield, Yield Curve, Wishart Process}

\begin{abstract}
We develop the HJM framework for forward rates driven by affine processes on the state space of symmetric positive semidefinite matrices. 
In this setting we find an explicit representation for the long-term yield in terms of the model parameters.
This generalises the results of \cite{Karoui97} and \cite{bh2012}, where the long-term yield is investigated under no-arbitrage 
assumptions in a HJM setting using Brownian motions and L\'{e}vy processes respectively. 
\end{abstract}
\maketitle


\section{Introduction}
Long-term interest rates are particularly relevant for the pricing and hedging of long-term fixed-income securities, pension funds, life and accident insurances, or interest rate swaps with a very long time to maturity. 
Thus, the modelling of long-term interest rates is the topic of several contributions which however do not provide a unique definition of long-term interest rates or yield. 
In Section \ref{Long_Term_Yield_in_an_Affine_HJM_Setting} we provide a brief discussion on the different conventions concerning the time to maturity defining the concept of long-term yield in the literature.
Several studies address the topic from a more mathematical or a more macroeconomic point of view. 
The macroeconomic approach is focused on identifying the macroeconomic factors influencing the long-term yield. 
For example the paper \cite{article_Mankiw} examines the impact of monetary and fiscal policies 
on long-term interest rates and rejects the hypothesis that long-term interest are overly sensitive to short-term rates. 
The article \cite{article_Guerkaynak} also studies the impact of macroeconomic news and monetary policy surprises on long-term yields and 
presents evidence that these factors have significant effects on short-term as well as on long-term interest rates. The work \cite{article_hoerdahl} describes a joint 
model of macroeconomic and yield curve dynamics where the continuously compounded spot rate is an affine function 
dependent on macroeconomic state variables. With the help of this model the influence of macroeconomic effects on the long-term yield can be measured. 
The finding of a model that jointly characterises the behaviour of the yield curve and macroeconomic variables as well as state results for the short-term and long-term 
interest rates is also the subject of \cite{article_Ang} and \cite{article_Diebold}. 
In \cite{article_Ang} a vector autoregression model is used to describe the relationship between interest rates and macroeconomy, whereas \cite{article_Diebold} uses 
a latent factor model with the inclusion of macroeconomic variables to model the yield curve. 
In \cite{article_Kim} the yield curve is modeled by a three-factor model, where the interest rates can be described with the help of three underlying latent factors which are employed in order to 
explain the empirical result of falling long-term yields.
\par
Mathematical approaches consider the long-term yield as an interest rate with time to maturity tending to infinity. 
In the textbook \cite{Book_Carmona} as well as in \cite{bh2012}, \cite{Karoui97}, and \cite{Yao2000}, the long-term yield is defined as the limit of the continuously compounded spot rate. 
In this paper we adopt this definition.
The respective form of the long-term yield then depends on the chosen interest rate model, whereas there can be made some universal statements 
concerning the asymptotic behaviour of yields in an arbitrage-free market, independent of the chosen setting.
\par
One of the most important results concerning the asymptotic behaviour of yields 
is that in an arbitrage-free market, long-term zero-coupon rates can never fall, as first stated in \cite{dybvig96}, consequently referred to as \emph{DIR-Theorem}. 
This result was made rigorous by \cite{article_McCulloch}. 
An alternative proof using a different definition of arbitrage can be found in \cite{article_Schulze}.
Then, \cite{article_Hubalek} provided a generalisation of the proof of the \emph{DIR-Theorem},
where they assume the existence of an equivalent martingale measure and omit some measurability conditions.
The assumption of the existence of an equivalent martingale measure is relaxed in \cite{article_Kardaras}.
Finally, \cite{article_Goldammer} generalised the theorem by dropping the requirement of the existence of the long-term yield, 
and showed that the limit superior of zero-coupon rates and forward rates never fall. 
From these results the general behviour of the long-term yield has been clarified. 
However only a few studies have contributed to find explicitly the form for the long-term yield in specific models.

Since it is very important to investigate the concrete structure of the long-term yield for several applications, 
our aim here is to provide an explicit representation of the long-term yield in an HJM framework driven by a general affine process 
on $S_{d}^{+}$. 
Concrete computations of the long-term yield as limit of the standard 
yield have been done in \cite{Book_Carmona}, \cite{Karoui97}, \cite{Yao1999}, \cite{Yao2000} in a Brownian motion setting 
and more recently in \cite{bh2012} in a general L\'{e}vy setting.
In \cite{bh2012} an explicit form for the long-term yield is provided that takes also in 
account the impact of jumps on the long-term behaviour. 
In this paper our setting presents the main advantage that the forward 
curve can be described by taking account of a rich interdependence 
structure among factors.
This provides a flexible way of describing the impact of different risk factors and of their correlations on the long-term yield. 
Under some integrability and measurability conditions on the parameters, we are able to obtain an explicit form of the long-term yield, which results to be independent of the underlying probability measure. 
This extends a result of Section 2.2 in \cite{Karoui97} to a multifactor setting including jumps. 
Moreover, we prove that in our context jumps in the dynamics of the yield do not impact the long-term behaviour.

In order to model the long-term yield, we first provide an extension of the classical Heath-Jarrow-Morton framework to a setting where the market is driven by semimartingale taking values on the cone $S_d^+$ of positive semidefinite symmetric $d\times d$ matrices. 
This class of stochastic processes has appealing features and is increasingly studied in finance research, in particular for modelling multivariate stochastic volatilities in equity and fixed income models, cf.\,\,e.g.\,\,\cite{article_Benabid}, \cite{article_DaFonseca2011}, \cite{article_DaFonseca4}, \cite{article_DaFonseca1}, \cite{article_gou03}, \cite{article_MPS02}, and \cite{article_Richter}. 
It allows to model a whole family of factors which share non-linear links among each other, providing a more realistic description of the market. In many situations, the presence of stochastic correlations among factors does not come at the cost of a loss of analytical tractability, as these processes are affine, in the sense of \cite{article_Cuchiero}. The class of affine processes on $S_{d}^{+}$, i.e. stochastically continuous Markov processes with the feature that the Laplace transform can be represented as an exponential-affine function, was 
introduced to applications in finance by \cite{article_gou06} and \cite{article_gou03} in the form of Wishart processes, a particular affine process first described by Bru in \cite{article_Bru}. 
Theoretical background to affine processes on $S_{d}^{+}$ can be found, among other publications, in \cite{phdthesis_Cuchiero}, \cite{article_Cuchiero}, \cite{article_Cuchiero2}, \cite{article_DFS}, \cite{article_Gnoatto}, \cite{article_Grasselli},  and \cite{article_MPS}. 
A first application of Wishart processes for short rates modelling is given in \cite{article_Gnoatto2}, while a Libor model using affine processes is constructed in \cite{article_FonsecaGnoatto}.  
Here we consider for the first time an affine HJM framework on $S_d^+$, where we develop formulas for forward rates, short rates, and continuously compounded spot rates as well as determine the HJM condition on the drift. Note also that we allow for general affine processes on $S_d^+$, i.e. we admit jumps. This setting provides a flexible and concise way of taking into account the influence of a large number of factor on interest rates dynamics  and represents a further contribution in capturing the dependence structure affecting the interest rates evolution. Moreover, in the final examples, we show that in this setting we can originate affine multidimensional realisations for the forward rate in the sense of \cite[Def. 3]{Chiarella_Kwon}.
\par
An interesting aspect of our study is the use of a matrix-valued driver, whose
elements are stochastically correlated among each other. Our choice for such 
rich multi-dimensional dynamics is open for different economic
interpretations which are beyond the scope of the present contribution. Let
us remark however that we view our specification as beneficial in two possible contexts. 
It provides an alternative way to capture the intrinsic multivariate and dynamic nature of the yield curve, and, if we
extend our view to the post-crisis interest rate market, i.e.\,\,to a multiple
curve interest rate setting, 
it can bes used in the description of positive spreads among different curves to take into account the impact of credit and liquidity/systemic risk.
Furthermore in the special case of a Wishart process as driving factor we are able to provide the correlation structure in a concise way, 
since the Wishart dynamics automatically guarantee that the elements of the driving process are stochastically correlated. 
\par
Explicit results on asymptotic behaviour of the long term yield are also of wide interest,
especially they could be relevant for the literature on long-term
risk, see in particular the recent contribution \cite{Hansen_Scheinkman2012}.
In \cite{Hansen_Scheinkman2012} the goal of the paper is to provide a term structure of risk prices by
changing the investment horizon. Our explicit results could be then used to provide analytical expression for the growth-rate risk in the limit.




\par
The paper is structured as follows. In Section \ref{Affine_Processes} we present the main properties of affine processes on $S_{d}^{+}$ as well as features that 
are important in the course of this paper. 
Then, in Section \ref{Affine_HJM_Framework} analytical expressions for different interest rates are developed under the HJM framework with an 
affine process $X$ on $S_{d}^{+}$ as stochastic driver of the forward rate.
In Section \ref{Long_Term_Yield_in_an_Affine_HJM_Setting} we provide an explicit representation of the long-term yield in the HJM framework on $S_{d}^{+}$ and consider some concrete examples.

\section{Affine Processes on $S^{+}_{d}$}\label{Affine_Processes}
Affine processes were initially studied by \cite{article_DuffieKan} and later fully characterised by \cite{article_DFS} on the state space 
$\mathbb{R}^{m}_{+} \times \mathbb{R}^{n}$ with $m,n \in \mathbb{N}$. The theoretical framework for affine processes on the state space $S_{d}^{+}$, 
can be found in extensive forms in \cite{article_Cuchiero} and \cite{article_Mayerhofer_Jumps}. 
In this section, we state, for the reader`s convenience, the results of these works which are used in the course of this paper as well as the basic required notations. 
In general, for the stochastic background and notation we refer to \cite{Book_Protter}.
Let $d \in \mathbb{N}$. Then, $\mathcal{M}_{d}$ 
denotes the set of all $d \times d$ matrices with entries in $\mathbb{R}$,  
$S_{d}$ is the space of symmetric $d \times d$ matrices with entries in $\mathbb{R}$,
and $S_{d}^{+}$ stands for the cone of symmetric $d \times d$ positive semidefinite matrices with entries in $\mathbb{R}$
which induces a partial order relation on $S_{d}$:
\begin{equation}
 \text{For } x, y \in S_{d} \text{ it is } x \preceq y \text{ if } y - x \in S_{d}^{+}\,.\nonumber
\end{equation}
The space $\mathcal{M}_{d}$ is endowed with the scalar product $A \cdot B \colonequals \operatorname{Tr}\left[A^{\top}\! B\right]$ 
for $A, B \in \mathcal{M}_{d}$, where $\operatorname{Tr}\left[A\right]$ denotes the trace of the matrix $A$. 
\par
Throughout this paper, given $A \subseteq \mathcal{M}_{d}$, $\mathcal{B}\!\left(A\right)$ denotes the Borel $\sigma$-algebra on $A$ 
and $b\!\left(A\right)$ the Banach space of bounded real-valued Borel-measurable functions $f$ on $A$ with norm $\left\Vert f \right\Vert_{\infty} = \sup_{x \in A} \left\vert f\!\left(x\right) \right\vert$. 
\par
Let $(\Omega, \mathcal{F} ,\left(\mathcal{F}_t \right)_{t \geq 0}, \mathbb{P}_{x})$ be a filtered probability space with the filtration $\left(\mathcal{F}_{t}\right)_{t \geq 0}$ 
satisfying the usual conditions of completeness and right-continuity and $X \colonequals \left(X_{t}\right)_{t \geq 0}$ a stochastic process on this probability space. 
For $x \in S_{d}^{+}$, $\mathbb{P}_{x}$ is a probability measure such that 
$\mathbb{P}_{x}\!\left(X_{0} = x\right) = 1$. 
Given $t > 0$, $X_{t-} \colonequals \lim_{s \uparrow t} X_{s}$, we define
\begin{equation}\label{Jump_Definition}
 \Delta X_{t} \colonequals X_{t} - X_{t-}\,,
\end{equation}
the jump at $t$, $\Delta X_{0} \equiv 0$.
\par
Next, we define the transition probabilities for all $t \geq 0$ as:
\begin{equation*}
 p_{t} : S^{+}_{d} \times \mathcal{B}\!\left(S^{+}_{d}\right) \rightarrow \left[0,1\right], \left(x,B\right) \mapsto \mathbb{P}_{x}\!\left(X_{t} \in B\right)\,.
\end{equation*}
Further, let $\left(P_{t}\right)_{t \geq 0}$ be a semigroup such that
\begin{equation}\label{Semigroup_Definition}
P_{t}f\!\left(x\right) \colonequals \int_{S^{+}_{d}\!} f\!\left(\xi\right)\, p_{t}\!\left(x,d\xi\right) = \mathbb{E}_{x}\!\left[f\!\left(X_{t}\right)\right], x \in S_{d}^{+},
\end{equation}
where $f \in b\!\left(S^{+}_{d}\right)$.
\par
We consider a time-homogeneous Markov process $X$ with state space $S^{+}_{d}$, i.e. the Markov property holds for all $A \in \mathcal{B}\!\left(S^{+}_{d}\right), x \in S^{+}_{d},$ and $s,t \geq 0$ (cf.\,\,Definition 17.3 in \cite{Book_Klenke}):
\begin{equation}\label{Markov_Property}
 \mathbb{P}_{x}\!\left(X_{t+s} \in A \mid \mathcal{F}_{s}\right) = p_{t}\!\left(X_{s},A\right)\ \ \  \mathbb{P}_{x}\text{ - a.s.}\nonumber
\end{equation}

Next, we want to define the characteristics of an affine process on $S_{d}^{+}$ (cf.\,\,Definition 2.1 in \cite{article_Cuchiero}).

\begin{definition}\label{Affine_Process_Definition}
 A Markov process $X$ with values in $S_{d}^{+}$ is called \emph{affine} if the following two properties hold:
\begin{enumerate}
  \item[(i)] It is stochastically continuous, i.e. it holds for all $t \geq 0$ and all $\epsilon > 0$: 
    \begin{equation}\label{Affine_Process_Definition_Equation_1}
      \lim_{s \rightarrow t} \mathbb{P}_{x}\!\left(\left\Vert X_{s} - X_{t}\right\Vert > \epsilon\right) = 0\,.\nonumber
    \end{equation}
  \item[(ii)] Its Laplace transform has exponential-affine dependence on the initial state, i.e. the following equation holds for all $t \geq 0$ and $u, x \in S_{d}^{+}$:
    \begin{equation}\label{Affine_Process_Definition_Equation_2}
      P_{t}e^{-\operatorname{Tr}\left[u x\right]} \overset{(\ref{Semigroup_Definition})}{=} \int_{S_{d}^{+}}\! {e^{-\operatorname{Tr}\left[u \xi \right]}\, p_{t}\!\left(x,d\xi\right)} = e^{-\phi\left(t,u\right) - \operatorname{Tr}\left[ \psi\left(t,u\right) x\right]}\,,
    \end{equation}
     for some functions $\phi: \mathbb{R}_{+} \times S_{d}^{+} \rightarrow \mathbb{R}_{+}$ and $\psi: \mathbb{R}_{+}\times S_{d}^{+} \rightarrow S_{d}^{+}$.
\end{enumerate}
\end{definition}


From the stochastic continuity of $X$ follows directly the weak convergence of the distributions $p_{t}\!\left(x,\cdot\right)$, $t \geq 0$, 
i.e. it holds for all $t \geq 0$ (cf.\,\,Satz 5.1 in \cite{Book_Bauer}):
\begin{equation*}
 \lim_{s \rightarrow t} p_{s}\!\left(x,\cdot\right) = p_{t}\!\left(x,\cdot\right)\,.
\end{equation*}

Note, that due to the non-negativity of $X$ the Laplace transform is well-defined and can be used to characterise an affine process. 
Further, in consequence of the stochastic continuity of the process according to Proposition 3.4 in \cite{article_Cuchiero}, the process 
$X$ is regular in the sense of Definition 2.2 in \cite{article_Cuchiero}.
\par
As well we consider that the affine Markov process is conservative, that means that the process will remain almost surely on the state space $S_{d}^{+}$ 
all the time.

\begin{definition}\label{Conservative_Affine_Process_Definition}
The affine process $X$ is called \emph{conservative} if for all $t \geq 0$ the following condition holds:
\begin{equation}\label{Conservative_Affine_Process_Definition_Equation_1}
p_{t}\!\left(x,S_{d}^{+}\right) = 1\,,\nonumber
\end{equation}
i.e. $X_{t} \in S_{d}^{+}$ $\mathbb{P}_{x}$-a.s.
\end{definition}

Now, we are able to introduce the so-called admissible parameter set which generalises the concept of L\'{e}vy 
triplet to the setting of affine processes on $S_{d}^{+}$ (cf.\,\,Definition 3.1 in \cite{article_Mayerhofer_Jumps}).

\begin{definition}\label{Admissible_Parameter_Set}
 An \emph{admissible parameter set} $\left(\alpha,b,B,m,\mu\right)$ consists of
  \begin{enumerate}
   \item[(i)] a linear diffusion coefficient $\alpha \in S_{d}^{+}$,
   \item[(ii)] a constant drift term $b \in S_{d}^{+}$ which satisfies
              \begin{equation}\label{Drift_Term_Equation}
               b \succeq \left(d-1\right) \alpha\,,\nonumber
              \end{equation}
   \item[(iii)] a Borel measure $m$ on $S_{d}^{+}\! \setminus\! \left\{ 0\right\}$ to represent the constant jump term
              \begin{equation}\label{Constant_Jump_Term_Equation}
               \int_{S_{d}^{+}\! \setminus \left\{ 0\right\}} \! \left(\left\Vert \xi \right\Vert \land 1\right) \, m\!\left(d\xi\right) < \infty\,,
              \end{equation}
   \item[(iv)] a linear jump coefficient $\mu : {S_{d}^{+}\! \setminus \left\{ 0\right\}} \rightarrow {S_{d}^{+}\! \setminus \left\{ 0\right\}}$ which is a $\sigma$-finite measure and satisfies 
              \begin{equation}\label{Linear_Jump_Coefficient_Equation}
               \int_{S_{d}^{+}\! \setminus \left\{ 0\right\}} \! \left(\left\Vert \xi \right\Vert \land 1\right) \, \mu\!\left(d\xi\right) < \infty\,,
              \end{equation}
   \item[(v)] a linear drift $B : S_{d}^{+} \rightarrow S_{d}^{+}$ that satisfies the condition
              \begin{equation}\label{Linear_Drift_Equation_1}
               \operatorname{Tr}\!\left[B\!\left(x\right) u\right] \geq 0 \text{ for all } x, u \in S_{d}^{+} \text{ with } \operatorname{Tr}\!\left[x u\right] = 0\,.\nonumber
              \end{equation}
  \end{enumerate}
\end{definition}


\begin{theorem}\label{Riccati_Equations_Theorem}
Suppose $X$ is a conservative affine process on $S_{d}^{+}$ with $d\geq 2$. 
Then $X$ is regular and has the Feller property. 
Moreover, there exists an admissible parameter set $\left(\alpha,b,B,m,\mu\right)$ such that 
$\phi: \mathbb{R}_{+} \times S_{d}^{+} \rightarrow \mathbb{R}_{+}$ and $\psi: \mathbb{R}_{+}\times S_{d}^{+} \rightarrow S_{d}^{+}$ 
in \eqref{Affine_Process_Definition_Equation_2} solve the generalised Riccati differential equations for $u \in S_{d}^{+}$
\begin{align}
 & \partial_{t} \phi\!\left(t,u\right) = F\!\left(\psi\!\left(t,u\right)\right),\ \ \ \phi\!\left(0,u\right) = 0\,, \label{Riccati_Equations_Theorem_Equation_1}\\
 & \partial_{t} \psi\!\left(t,u\right) = R\!\left(\psi\!\left(t,u\right)\right),\ \ \ \psi\!\left(0,u\right) = 0\,,  \label{Riccati_Equations_Theorem_Equation_2}
\end{align}
with
 \begin{align}
 & F\!\left(u\right) \colonequals \operatorname{Tr}\!\left[b u\right] - \int_{S_{d}^{+}\! \setminus \left\{ 0\right\}}\!\left(e^{-\operatorname{Tr}\left[u \xi\right]} - 1\right)\, m\!\left(d\xi\right)\,, \label{Riccati_Equations_Theorem_Equation_3}\\
 & R\!\left(u\right) \colonequals -2\, u\, \alpha\, u + B^{\top}\!\left(u\right)- \int_{S_{d}^{+}\! \setminus \left\{ 0\right\}}\!\left(e^{-\operatorname{Tr}\left[u \xi\right]} - 1\right)\, \mu\!\left(d\xi\right)\,. \label{Riccati_Equations_Theorem_Equation_4}
 \end{align}
Conversely, let $\left(\alpha,b,B,m,\mu\right)$ be an admissible parameter set and $d\geq 2$. 
Then there exists a unique conservative affine process $X$ on $S_{d}^{+}$ such that the affine property \eqref{Affine_Process_Definition_Equation_2} 
holds for all $t \geq 0$ and $u, x \in S_{d}^{+}$ with $\phi: \mathbb{R}_{+} \times S_{d}^{+} \rightarrow \mathbb{R}_{+}$ and $\psi: \mathbb{R}_{+}\times S_{d}^{+} \rightarrow S_{d}^{+}$ 
 given by \eqref{Riccati_Equations_Theorem_Equation_1} and \eqref{Riccati_Equations_Theorem_Equation_2}.
\end{theorem}

\begin{proof}
 Cf.\,\,Theorem 2.4 of \cite{article_Cuchiero} and Theorem 4.1 of \cite{article_Mayerhofer_Jumps}.
\end{proof}

Besides the admissible parameter set, we need to define the matrix variate Brownian motion for the representation of the affine process $X$ (cf.\,\,Definition 3.23 in \cite{article_Pfaffel}).

\begin{definition}\label{Matrix_Variate_Brownian_Motion_Definition}
 A \emph{matrix variate Brownian motion} $W \in \mathcal{M}_{d}$ is a matrix consisting of $d^2$ independent, one-dimensional 
Brownian motions $W_{ij}, 1 \leq i,j \leq d$.
\end{definition}

\begin{remark}\label{Finite_Variation_Jumps_Remark}
By \emph{(3.3)} of \cite{article_Mayerhofer_Jumps} we obtain that in the case of $d \geq 2$, the affine process $X$ has only jumps of finite variation, 
i.e for all $t\geq 0$
\begin{equation}\label{Finite_Variation_Jumps_Remark_Equation}
 \int_{0}^{t}\!\int_{S_{d}^{+}\! \setminus \left\{ 0\right\}}\! \left\Vert\xi\right\Vert \, \mu^{X\!}\!\left(ds,d\xi\right) < \infty\,.
\end{equation}
\end{remark}

Now, we can state the following representation of $X$.

\begin{theorem}\label{Affine_Process_Representation}
Let $X$ be a conservative affine process on $S_{d}^{+}$, $d\geq 2$, with admissible parameter set $\left(\alpha,b,B,m,\mu\right)$, 
where $Q \in \mathcal{M}_{d}$ such that $Q^{\top}Q = \alpha$. Then there exists a matrix Brownian motion 
$W \in \mathcal{M}_{d}$ such that $X$ admits the following representation:
\begin{equation}\label{Affine_Process_Representation_Equation_1}
X_{t} = x + \!\int_{0}^{t}\!\! \left(b\!+\!B\!\left(X_{s}\right)\right) ds + \!\int_{0}^{t}\!\!\left(\!\sqrt{X_{s}}dW_{s} Q\! +\! Q^{\top} dW_{s}^{\top}\! \sqrt{X_{s}}\right) +\! \int_{0}^{t}\!\int_{S_{d}^{+}\! \setminus \left\{ 0\right\}}\! \xi \, \mu^{X\!}\!\left(ds,d\xi\right), 
\end{equation}
where $\mu^{X\!}\!\left(ds,d\xi\right)$ is the random measure associated with the jumps of $X$, having the compensator
\begin{equation}\label{Affine_Process_Representation_Equation_2}
 \nu\!\left(dt,d\xi\right) \colonequals \left(m\!\left(d\xi\right) + \operatorname{Tr}\!\left[X_{t}\, \mu\!\left(d\xi\right)\right]\right) dt\,.
\end{equation}
\end{theorem}

\begin{proof}
 Cf.\,\,Theorem 3.4 in \cite{article_Mayerhofer_Jumps}.
\end{proof}

Note, that it is possible to choose $Q$ this way since
\begin{math}
Q^{\top}\!Q \in S_{d}^{+}
\end{math}
for all $Q \in \mathcal{M}_{d}$ due to Theorem 2.2 (ix) in \cite{article_Pfaffel}.

\begin{remark}\label{Wishart_Process_Remark}
If in \emph{Theorem \ref{Affine_Process_Representation}} we have $b = \delta \alpha$ with $\delta \geq 0$, ${B\!\left(z\right) = Mz + zM^{\top}}$ with $M \in \mathcal{M}_{d}$, and there are no jumps, the process $X$ is a Wishart process, 
cf.\,\,\cite{article_Bru}.
\end{remark}

Throughout this paper we consider $X$ to be a conservative, regular, affine process on the state space $S_{d}^{+}$ with $d \geq 2$, hence $X$ can be represented by equation 
(\ref{Affine_Process_Representation_Equation_1}). Furthermore, the linear drift coefficient $B$ is of the form
\begin{equation}\label{Linear_Drift_Equation_2}
 B\!\left(z\right) = Mz + zM^{\top} + G\!\left(z\right),\,z \in S_{d}^{+},
\end{equation}
where $M \in \mathcal{M}_{d}$ and $G : S_{d} \rightarrow S_{d}$ is linear satisfying $G\!\left(S_{d}^{+}\right) \subseteq S_{d}^{+}$ to encompass a wider range 
of affine processes (cf.\,\,(2.30) in \cite{article_Cuchiero}).
\par
 Note, that in the case of $X$ being not conservative, all subsequent calculations and the consequential results are still valid, as long as another set of 
admissible parameters is used with an additional constant killing rate term $c \in \mathbb{R}_{+}$ and an additional linear killing rate coefficient $\gamma \in S_{d}^{+}$. 
In the case of $d = 1$, the parameter set has to be extended by a truncation function for compensating the infinite variation part of the jumps. 
The most general admissible parameter set, encompassing the case of $X$ being not conservative on a state space with dimension $d = 1$, 
is stated in Definition 2.3 in \cite{article_Cuchiero}.



\section{Affine HJM Framework on $S_{d}^{+}$}\label{Affine_HJM_Framework}
 
We now provide a HJM framework to model the forward curve using affine processes on $S_{d}^{+}$ in the setting outlined in Section \ref{Affine_Processes}. 
\par
By a $T$-maturity zero-coupon bond we mean a contract that guarantees its holder the payment of one unit of currency at time $T$, with no intermediate payments. 
The contract value at time $t \leq T$ is denoted by $P\!\left(t,T\right)$ and the bond market satisfies the following hypotheses: (1) there exists a frictionless market for $T$-bonds for every maturity $T \geq 0$;
(2) $P\!\left(T,T\right) = 1$ for every $T \geq 0$; (3) for each fixed $t$, the zero-coupon bond price $P\!\left(t,T\right)$ is differentiable with respect to the maturity $T$. 
\par
The money market account is $\beta_{t} \colonequals \exp\!\left(\int_{0}^{t}\!r_{s}\,ds\right)$ with $r_{t}$ denoting the short rate at time $t$. 
We set $\Delta^2 \colonequals \left\{(t,T) \in \mathbb{R}_{+} \times \mathbb{R}_{+},\  t\leq T\right\}$ and assume the forward rates $f:\Omega \times \Delta^2 \rightarrow \mathbb{R}$ to evolve for every maturity $T > 0$ according to
\begin{align}\label{Forward_Rates_Process_Definition}
f\!\left(t,T\right) = f\!\left(0,T\right) + \int_{0}^{t}\! \alpha\!\left(s,T\right) \, ds + \int_{0}^{t} \operatorname{Tr}\!\left[\sigma\!\left(s,T\right) dX_{s} \right], \ 0\leq t\leq T,
\end{align}
where  $X$ is an affine conservative process with representation \eqref{Affine_Process_Representation_Equation_1} for a given initial value $x \in S_{d}^{+}$. 
Since we fix the initial value $X_{0} = x$, from now on we write $\mathbb{P}$ for $\mathbb{P}_{x}$.
We impose the following conditions on the drift $\alpha : \Omega \times \mathbb{R}_{+} \times \mathbb{R}_{+} \rightarrow \mathbb{R}$
and the volatilities $\sigma_{ij} : \Omega \times \mathbb{R}_{+} \times \mathbb{R}_{+} \rightarrow \mathbb{R}, \ i,j \in \left\{1,\dots,d\right\}$:\footnotemark

\footnotetext{For $\alpha$ and $\sigma$ we write the shortened version $\alpha\!\left(s,T\right) \colonequals \alpha\!\left(\omega,s,T\right)$ 
and $\sigma\!\left(s,T\right) \colonequals \sigma\!\left(\omega,s,T\right)$.}

\begin{assumption}\label{Assumption1}
\phantom{aa}\par\noindent
\begin{itemize}
\item $\alpha \colonequals \alpha\!\left(\omega ,s,u\right) : \left(\Omega\! \times \!\mathbb{R}_{+}\! \times \!\mathbb{R}_{+}, \mathcal{F}\! \otimes \mathcal{B}\!\left(\mathbb{R}_{+}\right)\! \otimes \mathcal{B}\!\left(\mathbb{R}_{+}\right) \right) \rightarrow \left(\mathbb{R}, \mathcal{B}\!\left(\mathbb{R}\right)\right)$ is jointly measurable.
\item For all $T \geq 0$: 
     \begin{equation}\nonumber
      \int_{0}^{T}\int_{0}^{T} \left\vert \alpha\!\left(s,u\right)\right\vert \, ds\, du < \infty\ \,\mathbb{P}\text{-a.s.}
     \end{equation}
\item For all $s, u \in \mathbb{R}_{+}$ and a.e. $\omega \in \Omega$: $\sigma\!\left(s,u\right) \in S_{d}^{+}$, i.e. $\sigma\!\left(s,u\right) $ is a symmetric positive semidefinite $d \times d$ matrix.
\item $\sigma_{ij} \colonequals \sigma_{ij}\!\left(\omega ,s,u\right) : \left(\Omega\! \times \!\mathbb{R}_{+}\! \times \!\mathbb{R}_{+}, \mathcal{F}\! \otimes \!\mathcal{B}\!\left(\mathbb{R}_{+}\right)\! \otimes \!\mathcal{B}\!\left(\mathbb{R}_{+}\right) \right)$ $\rightarrow \left(\mathbb{R}, \mathcal{B}\!\left(\mathbb{R}\right)\right)$ are jointly measurable for all $i,j \in \left\{1,\dots,d\right\}$.
\item For all $T \geq 0$: $\left(\alpha\!\left(s,T\right)\right)_{s \in \left[0,T\right]}$ and $\left(\sigma\!\left(s,T\right)\right)_{s \in \left[0,T\right]}$ are adapted.
\item For all $T \geq 0$: 
     \begin{equation}\nonumber
      \sup_{s,u \leq T} \left\Vert \sigma\!\left(s,u\right)\right\Vert < \infty\ \,\mathbb{P}\text{-a.s.}
     \end{equation}
\item For all $i,j \in \left\{1,\dots,d\right\}: \sigma_{ij} : \mathbb{R}_{+} \times \mathbb{R}_{+} \rightarrow \mathbb{R}$ is c\`{a}gl\`{a}d in both components.
\end{itemize}
\end{assumption}

Due to Assumption \ref{Assumption1} the forward rate process is well-defined in (\ref{Forward_Rates_Process_Definition}).
Note that other integrability conditions can be chosen to guarantee that the integrals in \eqref{Forward_Rates_Process_Definition} are well-defined. 
In this case the results of the paper will also apply under technical modifications of the proofs.

\begin{proposition}\label{Bond_Price_Process_Proposition}
If $X$ is a conservative affine process and \emph{Assumption \ref{Assumption1}} holds, then for every maturity $T > 0$ the 
zero-coupon bond price follows a process of the form
\begin{align}\label{Bond_Price_Process_Equation}
P\!\left(t,T\right) & = P\!\left(0,T\right) + \int_{0}^{t}\! P\!\left(s,T\right)\left(r_{s} + A\!\left(s,T\right)\right) \, ds \nonumber\\
& \phantom{===i} + 2 \int_{0}^{t}\! P\!\left(s,T\right) \operatorname{Tr}\!\left[\Sigma\!\left(s,T\right)\sqrt{X_{s}} \, dW_{s} Q\right] \nonumber\\
& \phantom{===i} + \int_{0}^{t}\! P\!\left(s-,T\right) \int_{S_{d}^{+}\! \setminus \left\{ 0\right\}}\! \left(e^{\operatorname{Tr}\left[ \Sigma\left(s,T\right)\,\xi\right]} - 1\right) \left(\mu^{X} - \nu\right)\left(ds,d\xi\right) ,
\end{align}
for $t \leq T$, where
\begin{equation}\label{Sigma_Definition}
 \Sigma\!\left(s,T\right) \colonequals - \int_{s}^{T}\! \sigma\!\left(s,u\right) \, du
\end{equation}
is the $T$-bond volatility and 
\begin{equation}\label{A_Definition}
A\!\left(t,T\right) \colonequals -\int_{t}^{T}\! \alpha\!\left(t,u\right) \, du - F\!\left(-\Sigma\!\left(t,T\right)\right) - \operatorname{Tr}\!\left[R\!\left(-\Sigma\!\left(t,T\right)\right)X_{t}\right],
\end{equation}
where $F$ and $R$ are given by \eqref{Riccati_Equations_Theorem_Equation_3}, \eqref{Riccati_Equations_Theorem_Equation_4} respectively.
\end{proposition}

\begin{proof}
The proof can be found in the appendix.
\end{proof}

Note that from Assumption \ref{Assumption1} it follows that $-\Sigma\!\left(t,T\right) \in S_{d}^{+}$ for all $t,T \geq 0$ 
since $\sigma\!\left(t,T\right) \in S_{d}^{+}$  and it is easy to show that $\int_{t}^{T}\!\sigma\!\left(t,u\right)\, du \in S_{d}^{+}$.
Therefore all necessary integrals are finite with respect to $\mu^{X}$, $\nu$, and the compensated jump measure $\left(\mu^{X}-\nu\right)$, 
since $X$ has jumps of finite variation and is regular due to Theorem \ref{Riccati_Equations_Theorem}. From this it also follows that 
$F\!\left(-\Sigma\!\left(t,T\right)\right)$ and $R\!\left(-\Sigma\!\left(t,T\right)\right)$ exist.

\begin{remark}\label{Remark_Bond_Price}
The bond-price process $P\!\left(t,T\right), 0\leq t\leq T$, can be rewritten the following way:
 \begin{align}\label{Remark_Bond_Price_Equation_1}
P\!\left(t,T\right) & = P\!\left(0,T\right)  + \int_{0}^{t}\! P\!\left(s,T\right)\left(r_{s} + C\!\left(s,T\right)\right) \, ds + \int_{0}^{t}\! P\!\left(s-,T\right) \operatorname{Tr}\!\left[\Sigma\!\left(s,T\right) dX_{s}\right]\nonumber\\
& \phantom{===} + \int_{0}^{t} \!\int_{S_{d}^{+}\! \setminus \left\{ 0\right\}\!}\! P\!\left(s-,T\right)\! \left(e^{\operatorname{Tr}\left[\Sigma\left(s,T\right)\, \xi\right]} \!-\! 1\! -\! \operatorname{Tr}\!\left[\Sigma\!\left(s,T\right) \xi\right]\right)\! \left(\mu^{X\!} - \nu\right)\!\left(ds,d\xi\right),
\end{align}
with for all $0\leq t\leq T$
\small
\begin{equation}\label{C_Definition}
 C\!\left(t,T\right) \colonequals A\!\left(t,T\right) - \operatorname{Tr}\!\left[\Sigma\!\left(t,T\right) \left(b + B\!\left(X_{t}\right)\right)\right] - \int_{S_{d}^{+}\! \setminus \left\{ 0\right\}\!}\! \operatorname{Tr}\!\left[\Sigma\!\left(t,T\right) \xi\right]\left(m\!\left(d\xi\right) + \operatorname{Tr}\!\left[X_{t} \mu\!\left(d\xi\right)\right]\right),
\end{equation}
\normalsize
where $A\!\left(t,T\right)$ is defined in \eqref{A_Definition}.
\end{remark}

\begin{proof}
Representation \eqref{Remark_Bond_Price_Equation_1} follows by \eqref{Bond_Price_Process_Equation}, \eqref{Affine_Process_Representation_Equation_1}, and \eqref{C_Definition}. 
\par
Note that due to Proposition 1.28 of Chapter II in \cite{Book_JacodShiryaev} we are able to combine the measures $\mu^{X\!}\!\left(ds,d\xi\right)$ and $\nu\!\left(ds,d\xi\right)$ to $\left(\mu^{X\!} - \nu\right)\!\left(ds,d\xi\right)$, 
since the affine process $X$ has only jumps of finite variation (cf.\,\,\eqref{Finite_Variation_Jumps_Remark_Equation}) 
and Assumption \ref{Assumption1} guarantees that all integrals above are finite.
\end{proof}

As an immediate consequence of representation \eqref{Bond_Price_Process_Equation} for the bond price, we obtain the following corollary.

\begin{corollary}
 For every maturity $T > 0$, the discounted zero-coupon bond price follows a process of the formula
\begin{align}\label{Discounted_Bond_Price_Process_Corollary}
 \frac{P\!\left(t,T\right)}{\beta_{t}} & = P\!\left(0,T\right) + \int_{0}^{t}\! \frac{P\!\left(s,T\right)}{\beta_{s}} A\!\left(s,T\right) ds + 2 \int_{0}^{t}\! \frac{P\!\left(s,T\right)}{\beta_{s}} \operatorname{Tr}\!\left[\Sigma\!\left(s,T\right)\!\sqrt{X_{s}} \, dW_{s} Q\right] \nonumber\\
& \phantom{===i} + \int_{0}^{t}\! \frac{P\!\left(s-,T\right)}{\beta_{s}} \int_{S_{d}^{+}\! \setminus \left\{ 0\right\}}\! \left(e^{\operatorname{Tr}\left[ \Sigma\left(s,T\right)\,\xi\right]} - 1\right) \left(\mu^{X} - \nu\right)\left(ds,d\xi\right) ,
\end{align}
for all $t\leq T$.
\end{corollary}

\begin{proof}
This follows directly from the definition of the money market account and Proposition \ref{Bond_Price_Process_Proposition}.
\end{proof}

We now investigate the restrictions on the dynamics \eqref{Forward_Rates_Process_Definition} under the assumption of no arbitrage. 
Let $\mathbb{Q} \sim \mathbb{P}$ be an equivalent probability measure. 
By Theorem 3.12 of \cite{article_bj_kab_run} there exists $\gamma \in \mathcal{M}_{d}$ with $\int_{0}^{t}\! \left\Vert \gamma_{s} \right\Vert^{2} \, ds < \infty$ for all $t \geq 0$ 
such that $W^{\ast}_{t} = W_{t} - \int_{0}^{t}\! \gamma_{s} \, ds$, $t \geq 0$, is a matrix variate Brownian motion under $\mathbb{Q}$ and
an $\mathcal{F}_{t} \otimes \mathcal{B}\!\left(\left[0,t\right]\right) \otimes \mathcal{B}\!\left(S_{d}^{+}\right)$ 
measurable function $K : \Omega \times \mathbb{R}_{+} \times S_{d}^{+}\! \setminus \left\{ 0\right\} \rightarrow \mathbb{R}_{+} $ with
\begin{equation}\label{Q_Compensator_Equation_1}
 \int_{0}^{t}\!\int_{S_{d}^{+}\! \setminus \left\{ 0\right\}} \! \left\vert K\!\left(s,\xi\right) \right\vert \, \nu\!\left(ds,d\xi\right) < \infty\ \, \mathbb{P}\text{-a.s.} \nonumber
\end{equation}
for all $t\geq 0$, such that $\mu^{X\!}$ has the $\mathbb{Q}$-compensator 
\begin{equation}\label{Q_Compensator_Equation_2}
 \nu^{\ast\!}\!\left(dt,d\xi\right) \colonequals K\!\left(t,\xi\right) \nu\!\left(dt,d\xi\right).
\end{equation}
Furthermore, for all $t \geq 0$
\begin{displaymath}
 \frac{d\mathbb{Q}}{d\mathbb{P}}\bigm|_{\mathcal{F}_{t}} = L_{t}
\end{displaymath}
with
\begin{align}\label{Radon}
 \log L_{t} & = \int_{0}^{t}\!\gamma_{s} \, dW_{s} - \int_{0}^{t}\! \left\Vert \gamma_{s} \right\Vert^{2} \, ds
 + \int_{0}^{t}\!\int_{S_{d}^{+}\! \setminus \left\{ 0\right\}}\! \log K\!\left(s,\xi\right) \mu^{X\!}\!\left(ds,d\xi\right) \nonumber\\
& \phantom{====} + \int_{0}^{t}\!\int_{S_{d}^{+}\! \setminus \left\{ 0\right\}}\! \left(1 - K\!\left(s,\xi\right)\right) \nu\!\left(ds,d\xi\right).
\end{align}

\begin{definition}\label{ELMM_Definition}
Let $\mathbb{Q} \sim \mathbb{P}$. Then $\mathbb{Q}$ is an \emph{equivalent local martingale measure (ELMM)} for the bond market if for all $T > 0$ the 
discounted bond price process $\frac{P\left(t,T\right)}{\beta_{t}}, t\in \left[0,T\right]$, is a $\mathbb{Q}$-local martingale.
\end{definition}

\begin{theorem}[\textbf{HJM drift condition on $S_{d}^{+}$}]\label{HJM_Drift_Condition_Theorem}
 A probability measure $\mathbb{Q} \sim \mathbb{P}$ with Radon-Nikodym density \eqref{Radon} is an \emph{ELMM} if and only if 
\small
\begin{align}\label{HJM_Drift_Condition_Equation}
  \alpha\!\left(t,T\right) & = - \operatorname{Tr}\!\left[\sigma\!\left(t,T\right)\left(b + B\!\left(X_{t}\right) + 2 \sqrt{X_{t}}\, \gamma_{t}\, Q\right)\right] - 4 \operatorname{Tr}\!\left[Q\, \sigma\!\left(t,T\right) X_{t}\, \Sigma\!\left(t,T\right) Q^{\top}\right] \nonumber\\
& \phantom{====} - \int_{S_{d}^{+}\! \setminus \left\{ 0\right\}\!}\! \operatorname{Tr}\left[\sigma\!\left(t,T\right) \xi\right] e^{\operatorname{Tr}\left[\Sigma\left(t,T\right)\,\xi\right]} K\!\left(t,\xi\right) \left(m\!\left(d\xi\right)\! +\! \operatorname{Tr}\!\left[X_{s} \mu\!\left(d\xi\right)\right]\right)
\end{align}
\normalsize
for all $T > 0,\, dt \otimes d\mathbb{P}$-a.s.\\
In this case, the $\mathbb{Q}$-dynamics of the forward rates $f\!\left(t,T\right), 0\leq t\leq T$, are of the form
\small
\begin{align}\label{Forward_Rate_Representation_Q}
f\!\left(t,T\right) & = f\!\left(0,T\right) + \int_{0}^{t}\! \left\{4 \operatorname{Tr}\left[ Q\, \sigma\!\left(s,T\right) X_{s}\, \int_{s}^{T}\!\sigma\!\left(s,u\right) du\  Q^{\top}\right] \right.\nonumber\\
& \phantom{=====} \left. - \int_{S_{d}^{+}\! \setminus \left\{ 0\right\}\!}\! K\!\left(s,\xi\right) \operatorname{Tr}\!\left[\sigma\!\left(s,T\right) \xi\right] \left(e^{\operatorname{Tr}\left[\Sigma\left(s,T\right)\,\xi\right]}\! -\! 1\right) \left(m\!\left(d\xi\right)\! +\! \operatorname{Tr}\!\left[X_{s} \mu\!\left(d\xi\right)\right] \right)\right\} ds\nonumber\\
& \phantom{=====} + \int_{0}^{t}\!\int_{S_{d}^{+}\! \setminus \left\{ 0\right\}}\!\operatorname{Tr}\!\left[\sigma\!\left(s,T\right) \xi\right] \left(\mu^{X\!}-\nu^{\ast\!}\right)\!\left(ds,d\xi\right)\nonumber\\
& \phantom{=====} + 2 \int_{0}^{t}\! \operatorname{Tr}\left[\sigma\!\left(s,T\right) \sqrt{X_{s}}\, dW^{\ast}_{s}\, Q\right]\,.
\end{align}
\normalsize
\end{theorem}

\begin{proof}
The proof can be found in the appendix.
\end{proof}

Theorem \ref{HJM_Drift_Condition_Theorem} shows that the important property of the classical HJM framework, established in \cite{hea92}, that the 
forward rates are only dependent on the volatility in an arbitrage-free market, still holds in the framework of affine processes on $S_{d}^{+}$.
\par
Next, we want to investigate how the short rate process $r_{t}, t\geq 0,$ can be represented in the current framework.

\begin{corollary}\label{Short_Rate_Process_Proposition}
 Suppose that $f\!\left(0,T\right)$, $\alpha\!\left(t,T\right)$ and $\sigma\!\left(t,T\right)$ are differentiable in $T$ for all $t \geq 0$, 
$\partial_{T} \alpha\!\left(t,T\right)$ is jointly measurable, adapted, and c\`{a}gl\`{a}d in $t$,
and $\partial_{T} \sigma\!\left(t,T\right)$ is jointly measurable, adapted, and c\`{a}gl\`{a}d in $t$. Further, it holds for all $t\geq 0$ that
\begin{equation}\label{Assumption_Forward_Rate_0_Derivative_Integrable}
 \int_{0}^{t}\! \left\vert \partial_{u} f\!\left(0,u\right) \right\vert \, du < \infty\,,\nonumber
\end{equation}
as well as
\begin{equation}\label{Assumption_Alpha_Derivative_Integrable}
 \int_{\mathbb{R}_{+}}\! \int_{\mathbb{R}_{+}}\! \left\vert \partial_{T} \alpha\!\left(t,T\right) \right\vert \, dt\, dT < \infty\,.\nonumber
\end{equation}
Then, the short rate process $\left(r_{t}\right)_{t\geq 0}$ is of the form
\begin{equation}\label{Short_Rate_Process_Equation}
 r_{t} = r_{0} + \int_{0}^{t}\! \phi\!\left(u\right) \, du + \int_{0}^{t}\!\operatorname{Tr}\!\left[\sigma\!\left(u,u\right) dX_{u}\right]\,,
\end{equation}
where
\begin{equation}\label{Phi_Definition}
\phi\!\left(u\right) \colonequals \alpha\!\left(u,u\right) + \partial_{u} f\!\left(0,u\right) + \int_{0}^{u}\! \partial_{u} \alpha\!\left(s,u\right)\, ds + \int_{0}^{u}\!\operatorname{Tr}\!\left[\partial_{u} \sigma\!\left(s,u\right) dX_{s}\right]\,.\nonumber
\end{equation}
\end{corollary}

\begin{proof}
Representation \eqref{Short_Rate_Process_Equation} is a consequence of the theorem of Fubini for integrable functions 
(cf.\,\,\cite{Book_Klenke}, Chapter 14, Theorem 14.16), the stochastic Fubini theorem (cf.\,\,\cite{Book_Protter}, Chapter IV, Theorem 65) 
and the characterisation of the short rate process that holds for all $s \geq 0$
\begin{align}\label{Short_Rate_Representation}
r_{u} & \colonequals f\!\left(u,u\right) \overset{\eqref{Forward_Rates_Process_Definition}}{=} f\!\left(0,u\right) + \int_{0}^{u}\! \alpha\!\left(s,u\right) \, ds + \int_{0}^{u}\! \operatorname{Tr}\!\left[ \sigma\!\left(s,u\right) dX_{s} \right]\,.
\end{align}
\end{proof}

We now calculate the yield process 
\begin{equation}\label{Yield_Process_Definition}
 Y\!\left(t,T\right) \colonequals - \frac{\log P\left(t,T\right)}{T-t},\, 0\leq t\leq T,
\end{equation}
for $T > 0$ in the HJM framework for affine processes on $S_{d}^{+}$. 
We recall that the term ``yield curve'' is used differently in the literature. For example, in \cite{brigo06} it is a combination of simply compounded spot rates for maturities up to 
one year and annually compounded spot rates for maturities greater than one year. In this paper we will refer to the function $T \mapsto Y\!\left(t,T\right)$ as yield curve in $t$, see also Section 2.4.4 of \cite{Book_Filipovic}.
\par
Note that if $f : \mathbb{R}^{n} \rightarrow S_{d}$ for some $n,d \in \mathbb{N}$, then it is for $a, b \in \mathbb{R},\, x \in \mathbb{R}^{n}$:
\begin{equation}\label{Trace_Function_Derivative_Equation}
 \operatorname{Tr}\!\left[\int_{a}^{b} f\!\left(x\right)^{\top} \partial_{x} f\!\left(x\right)\, dx\right] = 
\frac{1}{2} \left(\left\Vert f\!\left(b\right) \right\Vert^{2} - \left\Vert f\!\left(a\right) \right\Vert^{2}\right)\,.
\end{equation}

\begin{lemma}\label{Cont_Comp_Spot_Rate_Lemma}
Let $0 \leq t < T$ and let $X$ be an affine process as in \eqref{Affine_Process_Representation_Equation_1}. Under \emph{Assumption \ref{Assumption1}} and the \emph{ELMM} $\mathbb{Q}$
 the yield for $\left[t,T\right]$ can be expressed in the compact form
\begin{align}\label{Cont_Comp_Spot_Rate_Equation}
Y\!\left(t,T\right) & = Y\!\left(0;t,T\right) + 2 \int_{0}^{t} \! \operatorname{Tr}\!\left[Q \frac{\Gamma\!\left(s,T\right) - \Gamma\!\left(s,t\right)}{T-t} Q^{\top}\right]ds\nonumber\\
& \phantom{====} + \int_{0}^{t}\! \int_{S_{d}^{+}\! \setminus \left\{ 0\right\}}\!\frac{e^{\operatorname{Tr}\left[\Sigma\left(s,T\right)\, \xi\right]} - e^{\operatorname{Tr}\left[\Sigma\left(s,t\right)\, \xi\right]}}{T-t}\,\nu^{\ast\!}\!\left(ds,d\xi\right)\nonumber\\
& \phantom{====} - \int_{0}^{t}\! \int_{S_{d}^{+}\! \setminus \left\{ 0\right\}}\! \frac{\operatorname{Tr}\left[\left(\Sigma\!\left(s,T\right) - \Sigma\!\left(s,t\right)\right)\xi\right]}{T - t}\, \mu^{X\!}\!\left(ds,d\xi\right)\nonumber\\
& \phantom{====} - 2 \int_{0}^{t} \! \operatorname{Tr}\!\left[\frac{\Sigma\!\left(s,T\right) - \Sigma\!\left(s,t\right)}{T-t} \sqrt{X_{s}} dW_{s}^{\ast} Q\right]
\end{align}
with the continuously compounded forward rate for $\left[t,T\right]$ prevailing at $0$ given by
\begin{equation}\label{yield_0_equation}
 Y\!\left(0;t,T\right) \colonequals \frac{1}{T-t}\left(\int_{t}^{T}\!f\!\left(0,u\right) du\right)
 \end{equation}
and
\begin{equation}\label{Definition_Gamma}
 \Gamma\!\left(s,t\right) \colonequals \Sigma\!\left(s,t\right)X_{s}\Sigma\!\left(s,t\right)
\end{equation}
for all $s,t \geq 0$.
\end{lemma}

\begin{proof}
Let $0 \leq t < T$. 
Note that for $0\leq s \leq t\leq T$ it holds
\begin{align}\label{Help_Equation_Lemma_Cont_Comp_Spot_1}
 \int_{t}^{T} \! \sigma\!\left(s,u\right)\, du 
& \overset{(\ref{Sigma_Definition})}{=} - \left(\Sigma\!\left(s,T\right) - \Sigma\!\left(s,t\right)\right)\,.
\end{align}
Further, for some $a,b,s \geq 0$ it is
\begin{align}\label{Help_Equation_Lemma_Cont_Comp_Spot_2}
& \int_{a}^{b} \! \operatorname{Tr}\!\left[Q\, \sigma\!\left(s,u\right) X_{s}\, \Sigma\!\left(s,u\right)\, Q^{\top}\right] du \nonumber\\
& \phantom{===}\underset{\phantom{(\ref{Trace_Function_Derivative_Equation})}}{\overset{(\ref{Sigma_Definition})}{=}} - \int_{a}^{b} \! \operatorname{Tr}\!\left[Q\, \partial_{u}\Sigma\!\left(s,u\right) X_{s}\, \Sigma\!\left(s,u\right)\, Q^{\top}\right] du\nonumber\\
& \phantom{===}\overset{\phantom{(\ref{Trace_Function_Derivative_Equation})}}{=} - \int_{a}^{b} \! \operatorname{Tr}\!\left[\left(Q\, \Sigma\!\left(s,u\right) \sqrt{X_{s}}\right)^{\top}\! \partial_{u}\!\left(\!Q\, \Sigma\!\left(s,u\right) \sqrt{X_{s}}\right)\right] du\nonumber\\
& \phantom{===}\overset{(\ref{Trace_Function_Derivative_Equation})}{=} - \frac{1}{2}\left(\left\Vert Q\, \Sigma\!\left(s,b\right)\sqrt{X_{s}}\right\Vert^{2} - \left\Vert Q\, \Sigma\!\left(s,a\right)\sqrt{X_{s}}\right\Vert^{2}\right)\nonumber\\
& \phantom{===}\underset{\phantom{(\ref{Trace_Function_Derivative_Equation})}}{\overset{(\ref{Definition_Gamma})}{=}} - \frac{1}{2}\operatorname{Tr}\!\left[Q\left(\Gamma\!\left(s,b\right) - \Gamma\!\left(s,a\right)\right)Q^{\top}\right]\,.
\end{align}
Then by applying the Fubini theorems, the yield for $\left[t,T\right]$ is 
\begin{align}
 Y\!\left(t,\!T\right) & \overset{\phantom{\eqref{Forward_Rate_Representation_Q}}}{=} \frac{1}{T-t}\left(\int_{t}^{T}\!f\!\left(t,u\right) du\right)\nonumber\\
& \overset{(\ref{Forward_Rate_Representation_Q})}{=} \int_{t}^{T}\frac{f\!\left(0,\!u\right)}{T-t}\, du - \frac{4}{T-t} \int_{t}^{T}\! \int_{0}^{t}\! \operatorname{Tr}\!\left[Q\, \sigma\!\left(s,u\right) X_{s} \Sigma\!\left(s,u\right) Q^{\top}\right] ds \, du\nonumber\\
& \phantom{====} + \frac{1}{T\!-\!t}\int_{t}^{T}\! \int_{0}^{t}\!\int_{S_{d}^{+}\! \setminus \left\{ 0\right\}}\!\operatorname{Tr}\!\left[\sigma\!\left(s,u\right) \xi\right] \left(\mu^{X\!} - \nu^{\ast\!}\right)\!\left(ds,d\xi\right)\, du\nonumber\\
& \phantom{====} - \frac{1}{T\!-\!t}\int_{t}^{T}\! \int_{0}^{t}\! \int_{S_{d}^{+}\! \setminus \left\{ 0\right\}}\! \operatorname{Tr}\!\left[\sigma\!\left(s,T\right) \xi\right]\left(e^{\operatorname{Tr}\left[\Sigma\left(s,T\right)\,\xi\right]} - 1\right) \nu^{\ast\!}\!\left(ds,d\xi\right)\, du\nonumber\\
& \phantom{====} + \frac{2}{T\!-\!t} \int_{t}^{T}\! \int_{0}^{t} \operatorname{Tr}\!\left[\sigma\!\left(s,u\right) \sqrt{X_{s}}\, dW_{s}^{\ast} Q\right] du\nonumber\\
& \overset{\eqref{yield_0_equation}}{=} Y\!\left(0;t,T\right) - \frac{4}{T-t} \int_{0}^{t}\! \int_{t}^{T} \operatorname{Tr}\!\left[Q\, \sigma\!\left(s,u\right) X_{s} \Sigma\!\left(s,u\right) Q^{\top}\right] du \, ds\nonumber\\
& \phantom{====} - \frac{1}{T\!-\!t}\int_{0}^{t}\! \int_{t}^{T}\!\int_{S_{d}^{+}\! \setminus \left\{ 0\right\}}\!\partial_{u} \operatorname{Tr}\!\left[\Sigma\!\left(s,u\right)\xi\right]\left(\mu^{X\!} - \nu^{\ast\!}\right)\!\left(du,d\xi\right)\, ds\nonumber\\
& \phantom{====} + \frac{1}{T\!-\!t}\int_{0}^{t}\! \int_{t}^{T}\!\int_{S_{d}^{+}\! \setminus \left\{ 0\right\}}\! \partial_{u} e^{\operatorname{Tr}\left[\Sigma\left(s,u\right)\,\xi\right]}\,\nu^{\ast\!}\!\left(du,d\xi\right)\,ds \nonumber\\
& \phantom{====} - \frac{1}{T\!-\!t}\int_{0}^{t}\! \int_{t}^{T}\!\int_{S_{d}^{+}\! \setminus \left\{ 0\right\}}\! \partial_{u}\operatorname{Tr}\!\left[\Sigma\!\left(s,u\right)\xi\right]\, \nu^{\ast\!}\!\left(du,d\xi\right)\,ds \nonumber\\
& \phantom{====}  + \frac{2}{T\!-\!t} \int_{0}^{t}\! \operatorname{Tr}\!\left[ \int_{t}^{T} \! \sigma\!\left(s,u\right) du \, \sqrt{X_{s}}\, dW_{s}^{\ast} Q\right]\nonumber\\ 
& \overset{\eqref{Help_Equation_Lemma_Cont_Comp_Spot_1}}{\underset{(\ref{Help_Equation_Lemma_Cont_Comp_Spot_2})}{=}} Y\!\left(0;t,T\right) + 2 \int_{0}^{t} \! \operatorname{Tr}\!\left[Q \frac{\Gamma\!\left(s,T\right) - \Gamma\!\left(s,t\right)}{T-t} Q^{\top}\right]ds \nonumber\\
& \phantom{====} + \int_{0}^{t}\! \int_{S_{d}^{+}\! \setminus \left\{ 0\right\}}\! \frac{e^{\operatorname{Tr}\left[\Sigma\left(s,T\right)\, \xi\right]} - e^{\operatorname{Tr}\left[\Sigma\left(s,t\right)\, \xi\right]}}{T-t}\, \nu^{\ast\!}\!\left(ds,d\xi\right)\nonumber\\
& \phantom{====} - \int_{0}^{t}\! \int_{S_{d}^{+}\! \setminus \left\{ 0\right\}}\! \frac{\operatorname{Tr}\left[\left(\Sigma\!\left(s,T\right) - \Sigma\!\left(s,t\right)\right)\xi\right]}{T - t}\, \mu^{X\!}\!\left(ds,d\xi\right)\nonumber\\
& \phantom{====} - 2 \int_{0}^{t} \! \operatorname{Tr}\!\left[\frac{\Sigma\!\left(s,T\right) - \Sigma\!\left(s,t\right)}{T-t} \sqrt{X_{s}}\, dW_{s}^{\ast} Q\right]\,.\nonumber
\end{align}
\end{proof}

\begin{corollary}\label{Remark_Yield}
 By \eqref{Affine_Process_Representation_Equation_2}, \eqref{Q_Compensator_Equation_2}, and \eqref{Cont_Comp_Spot_Rate_Equation} we obtain that
\begin{align}\label{Remark_Yield_Equation_0}
Y\!\left(t,T\right) & = Y\!\left(0;t,T\right) + \int_{0}^{t} \!\left\{ 2 \operatorname{Tr}\!\left[Q \frac{\Gamma\!\left(s,T\right) - \Gamma\!\left(s,t\right)}{T-t} Q^{\top}\right] \right.\nonumber\\
& \phantom{====} \left. + \int_{S_{d}^{+}\! \setminus \left\{ 0\right\}}\!\frac{M\!\left(s,t,T,\xi\right) K\!\left(s,\xi\right)}{T-t} \left(m\!\left(d\xi\right) + \operatorname{Tr}\!\left[X_{s}\mu\!\left(d\xi\right)\right]\right) \right\} ds\nonumber\\
& \phantom{====} - \int_{0}^{t}\! \int_{S_{d}^{+}\! \setminus \left\{ 0\right\}}\! \frac{\operatorname{Tr}\left[\left(\Sigma\!\left(s,T\right) - \Sigma\!\left(s,t\right)\right)\xi\right]}{T - t}\left(\mu^{X\!}-\nu^{\ast\!}\right)\!\left(ds,d\xi\right)\nonumber\\
& \phantom{====} - 2 \int_{0}^{t} \! \operatorname{Tr}\!\left[\frac{\Sigma\!\left(s,T\right) - \Sigma\!\left(s,t\right)}{T-t} \sqrt{X_{s}} dW_{s}^{\ast} Q\right]
\end{align}
with 
\begin{equation}\label{Definition_M}
 M\!\left(s,t,T,\xi\right) \colonequals e^{\operatorname{Tr}\left[\Sigma\left(s,T\right)\, \xi\right]} - e^{\operatorname{Tr}\left[\Sigma\left(s,t\right)\xi\right]} 
- \operatorname{Tr}\left[\left(\Sigma\!\left(s,T\right) - \Sigma\!\left(s,t\right)\right)\xi\right].
\end{equation}
\end{corollary}

\begin{proof}
 Let $0 \leq t < T$. Then, we have 
\small
\begin{align*}
Y\!\left(t,T\right) & \overset{\eqref{Cont_Comp_Spot_Rate_Equation}}{=} Y\!\left(0;t,T\right) + 2 \int_{0}^{t} \! \operatorname{Tr}\!\left[Q \frac{\Gamma\!\left(s,T\right) - \Gamma\!\left(s,t\right)}{T-t} Q^{\top}\right]ds\nonumber\\
& \phantom{====} + \int_{0}^{t}\! \int_{S_{d}^{+}\! \setminus \left\{ 0\right\}}\!\frac{e^{\operatorname{Tr}\left[\Sigma\left(s,T\right)\, \xi\right]} - e^{\operatorname{Tr}\left[\Sigma\left(s,t\right)\, \xi\right]}}{T-t}\,\nu^{\ast\!}\!\left(ds,d\xi\right)\nonumber\\
& \phantom{====} - \int_{0}^{t}\! \int_{S_{d}^{+}\! \setminus \left\{ 0\right\}}\! \frac{\operatorname{Tr}\left[\left(\Sigma\!\left(s,T\right) - \Sigma\!\left(s,t\right)\right)\xi\right]}{T - t}\, \mu^{X\!}\!\left(ds,d\xi\right)\nonumber\\
& \phantom{====} - 2 \int_{0}^{t} \! \operatorname{Tr}\!\left[\frac{\Sigma\!\left(s,T\right) - \Sigma\!\left(s,t\right)}{T-t} \sqrt{X_{s}} dW_{s}^{\ast} Q\right]\nonumber\\
& \overset{\phantom{\eqref{Cont_Comp_Spot_Rate_Equation}}}{=} Y\!\left(0;t,T\right) + 2 \int_{0}^{t} \! \operatorname{Tr}\!\left[Q \frac{\Gamma\!\left(s,T\right) - \Gamma\!\left(s,t\right)}{T-t} Q^{\top}\right]ds\nonumber\\
& \phantom{====} + \int_{0}^{t}\! \int_{S_{d}^{+}\! \setminus \left\{ 0\right\}}\!\frac{e^{\operatorname{Tr}\left[\Sigma\left(s,T\right)\, \xi\right]} - e^{\operatorname{Tr}\left[\Sigma\left(s,t\right)\, \xi\right]} - \operatorname{Tr}\left[\left(\Sigma\!\left(s,T\right) - \Sigma\!\left(s,t\right)\right)\xi\right]}{T-t}\, \nu^{\ast\!}\!\left(ds,d\xi\right)\nonumber\\
& \phantom{====} - \int_{0}^{t}\! \int_{S_{d}^{+}\! \setminus \left\{ 0\right\}}\! \frac{\operatorname{Tr}\left[\left(\Sigma\!\left(s,T\right) - \Sigma\!\left(s,t\right)\right)\xi\right]}{T - t}\left(\mu^{X\!}-\nu^{\ast\!}\right)\!\left(ds,d\xi\right)\nonumber\\
& \phantom{====} - 2 \int_{0}^{t} \! \operatorname{Tr}\!\left[\frac{\Sigma\!\left(s,T\right) - \Sigma\!\left(s,t\right)}{T-t} \sqrt{X_{s}} dW_{s}^{\ast} Q\right]\nonumber\\
& \underset{\eqref{Definition_M}}{\overset{\eqref{Affine_Process_Representation_Equation_2}}{=}}Y\!\left(0;t,T\right) +  \int_{0}^{t} \!\left\{ 2 \operatorname{Tr}\!\left[Q \frac{\Gamma\!\left(s,T\right) - \Gamma\!\left(s,t\right)}{T-t} Q^{\top}\right] \right.\nonumber\\
& \phantom{====} \left. + \int_{S_{d}^{+}\! \setminus \left\{ 0\right\}}\!\frac{M\!\left(s,t,T,\xi\right) K\!\left(s,\xi\right)}{T-t} \left(m\!\left(d\xi\right) + \operatorname{Tr}\!\left[X_{s}\mu\!\left(d\xi\right)\right]\right) \right\} ds\nonumber\\
& \phantom{====} - \int_{0}^{t}\! \int_{S_{d}^{+}\! \setminus \left\{ 0\right\}}\! \frac{\operatorname{Tr}\left[\left(\Sigma\!\left(s,T\right) - \Sigma\!\left(s,t\right)\right)\xi\right]}{T - t}\left(\mu^{X\!}-\nu^{\ast\!}\right)\!\left(ds,d\xi\right)\nonumber\\
& \phantom{====} - 2 \int_{0}^{t} \! \operatorname{Tr}\!\left[\frac{\Sigma\!\left(s,T\right) - \Sigma\!\left(s,t\right)}{T-t} \sqrt{X_{s}} dW_{s}^{\ast} Q\right].
\end{align*}
\normalsize
\end{proof}

\section{Long-Term Yield in an Affine HJM Setting on $S_{d}^{+}$}\label{Long_Term_Yield_in_an_Affine_HJM_Setting}
The expression ``long-term yield'' is subject to different interpretations in the literature. For instance, the European Central Bank 
understands the market yields of government bonds with time to maturity close to 10 years as long-term interest rates (cf.\,\,\cite{ECB2013}), 
whereas in \cite{Shiller79} also high-grade bonds with time to maturity longer than 20 years are examined to investigate 
long-term yields. In \cite{Yao2000} it is pointed out that for the valuation of some financial securities yield curves with maturities up to 100 years 
are necessary. Here we interpret ``long-term yield`` as the yield with time to maturity going to infinity. This approach, 
adopted by \cite{bh2012}, \cite{Book_Carmona}, \cite{dybvig96}, \cite{Karoui97}, is useful for modelling 
interest rates within a long-time horizon because the asymptotic behaviour can give information about the shape of the yield curve in the 
long run where only few empirical data is available. Here we study the asymptotic behaviour of the 
long-term yield in the affine HJM setting, introduced in Section \ref{Affine_HJM_Framework}.
\par
Throughout this section in the setting outlined in Section \ref{Affine_Processes} we assume directly that $\mathbb{P}$ is an 
ELMM for $\frac{P\left(t,T\right)}{\beta_{t}}$, $t \in \left[0,T\right]$, for all $T > 0$. 
More precisely, $X$ is a conservative affine process on $S_{d}^{+}$, $d \geq 2$, with representation \eqref{Affine_Process_Representation_Equation_1}, \eqref{Affine_Process_Representation_Equation_2}
 on $(\Omega, \mathcal{F} ,\left(\mathcal{F}_t \right)_{t \geq 0}, \mathbb{P})$ and the yield takes the form \eqref{Remark_Yield_Equation_0}, 
where we write $\nu$ instead of $\nu^{\ast}$ for the sake of simplicity.


\begin{assumption}\label{Assumption2}
Let $\Sigma\!\left(s,t\right)$ be defined as in \eqref{Sigma_Definition} for all $0 \leq s \leq t$ and $W$ a matrix variate Brownian motion.
There exists a progressively measurable process $w \in L\!\left(W\right)$ with values in $S_{d}^{+}$ such that 
 for every $i,j \in \left\{1,\dots,d\right\}$,  
$w_{ij}$ is a c\`{a}dl\`{a}g process
with 
\begin{equation}\label{Assumption2_w}
 \frac{1}{\sqrt{t}} \left\vert \Sigma\!\left(s,t\right)_{ij} \right\vert \leq w_{ij}\!\left(s\right)\ \,\mathbb{P}\text{-a.s.}
\end{equation}
for all $0 \leq s \leq t$ and $t\neq 0$.
\end{assumption}

\begin{definition}\label{LongTermYield}
 The \emph{long-term yield} $\left(\ell_{t}\right)_{t \geq 0}$ is the process defined by
\begin{equation}\label{LongTermYieldEqu}
 \ell_{t} \colonequals \lim\limits_{T \rightarrow \infty} Y\!\left(t,T\right),
\end{equation}
where $Y\!\left(t,T\right), t \in \left[0,T\right]$, is the yield process for $T \geq 0$ given by equation \eqref{Yield_Process_Definition}.
\end{definition}

\begin{definition}\label{LongTermDrift}
If the forward rate process is defined as in \eqref{Forward_Rates_Process_Definition}, the \emph{long-term drift} $\mu_{\infty}\!\left(t\right), t \geq 0$, is the process on $\mathcal{M}_{d}$ given by
\begin{equation}\label{LongTermDriftEqu}
 \mu_{\infty}\!\left(t\right) \colonequals \lim\limits_{T \rightarrow \infty} \frac{\Gamma\!\left(t,T\right)}{T-t} = \lim\limits_{T \rightarrow \infty} \frac{\Gamma\!\left(t,T\right)}{T}\ \,\mathbb{P}\text{-a.s.}
\end{equation}
for all $t\geq 0$, where $\Gamma\!\left(t,T\right), t \in \left[0,T\right]$, is introduced in \eqref{Definition_Gamma} for every $T \geq 0$.
\par
\noindent
Furthermore, the \emph{long-term volatility} $\sigma_{\infty}\!\left(t\right), t \geq 0$, is the process on $\mathcal{M}_{d}$ given by
\begin{equation}\label{LongTermVolaEqu}
 \sigma_{\infty}\!\left(t\right) \colonequals \lim\limits_{T \rightarrow \infty} \frac{\Sigma\!\left(t,T\right)}{T-t} = \lim\limits_{T \rightarrow \infty} \frac{\Sigma\!\left(t,T\right)}{T}\ \,\mathbb{P}\text{-a.s.}
\end{equation}
for all $t\geq 0$, where $\Sigma\!\left(t,T\right), t \in \left[0,T\right]$, is introduced in \eqref{Sigma_Definition} for every $T \geq 0$.
\end{definition}

Here we are supposing that the limits \eqref{LongTermYieldEqu}, \eqref{LongTermDriftEqu} and \eqref{LongTermVolaEqu} are well-defined.
The long-term yield can be characterised as an integral of $\mu_{\infty}$ and $\sigma_{\infty}$ by using the 
following results.

\begin{proposition}\label{LongTermYield0}
 Let $0 \leq t \leq T$. The \emph{long-term yield} at $0$ is
\begin{equation}\label{ContForwardPrev0}
 \lim\limits_{T \rightarrow \infty} Y\!\left(0;t,T\right) = \lim\limits_{T \rightarrow \infty} Y\!\left(0,T\right) = \ell_{0}\ \,\mathbb{P}\text{-a.s.}\nonumber
\end{equation}
\end{proposition}

\begin{proof}
 Cf.\,\,Proposition 3.3 of \cite{bh2012}.
\end{proof}

\begin{proposition}\label{LongTermVola0}
Under \emph{Assumption \ref{Assumption1}} and \emph{\ref{Assumption2}}, it holds for all $t \geq 0$:
\begin{equation}\label{Prop_LongTermVola_Equ}
 \lim\limits_{T \rightarrow \infty} 2\,\int\limits_{0}^{t}\! \operatorname{Tr}\!\left[\frac{\Sigma\!\left(s,T\right) - \Sigma\!\left(s,t\right)}{T-t} \sqrt{X_{s}}\, dW_{s}\, Q\right] = 2 \int\limits_{0}^{t} \!  \operatorname{Tr}\!\left[\sigma_{\infty}\!\left(s\right) \sqrt{X_{s}}\, dW_{s}\,Q\right],
\end{equation}
where $\sigma_{\infty}\!\left(s\right), s \geq 0$, is the long-term volatility process defined by equation \eqref{LongTermVolaEqu}, $\Sigma\!\left(s,t\right)$, $s \geq 0$, is defined for all $t \geq 0$ as in \eqref{Sigma_Definition}, and the convergence in \eqref{Prop_LongTermVola_Equ} is uniform on compacts in probability (ucp).
\end{proposition}

\begin{proof}
Fix $t \geq 0$.
By Assumption \ref{Assumption1} we have that for all compact intervals $\left[a,b\right]$ with $0 \leq a < b$
\begin{equation}\label{Prop_LongTermVol_1}
\sup_{t \in \left[a,b\right]} \left\vert \int_{0}^{t}\! \operatorname{Tr}\!\left[2\, Q\, \Sigma\!\left(s,t\right)\sqrt{X_{s}}\, dW_{s}\right]\right\vert
< \infty \ \,\mathbb{P}\text{-a.s.}\nonumber
\end{equation}
Consequently on every compact interval $\left[a,b\right]$
\begin{equation}\label{Prop_LongTermVol_1_1}
 \frac{1}{T} \sup_{t \in \left[a,b\right]}  \int_{0}^{t}\! \operatorname{Tr}\!\left[2\, Q\, \Sigma\!\left(s,t\right) \sqrt{X_{s}}\, dW_{s}\right] \overset{T \rightarrow \infty}{\longrightarrow} 0 \ \,\mathbb{P}\text{-a.s.}\nonumber
\end{equation}
Therefore
\begin{equation}\label{Prop_LongTermVol_2}
 \frac{1}{T} \int_{0}^{t}\! \operatorname{Tr}\!\left[2\, Q\, \Sigma\!\left(s,t\right) \sqrt{X_{s}}\, dW_{s}\right] \overset{T \rightarrow \infty}{\longrightarrow} 0 \text{ in ucp.}
\end{equation}
Next, we define $H^{T} \colonequals H^{T}_{s}, s \geq 0$, with 
\begin{equation}\label{Prop_LongTermVol_3}
 H^{T}_{s} \colonequals 2 \frac{Q\, \Sigma\!\left(s,T\right) \sqrt{X_{s}}}{T}\,.
\end{equation}
Then for $T \rightarrow \infty:\ H_{s}^{T} \rightarrow 2\,Q\,\sigma_{\infty}\!\left(s\right)\sqrt{X_{s}}\ \text{a.s. for all}\, s \geq 0$.
\par
\noindent
Since we investigate long-term interest rates it is sufficient to impose long times of maturity, say $T\geq 1$.
Due to Assumption \ref{Assumption2}, we then have that for all $0 \leq s \leq T$ with $T \geq 1$
\begin{align}\label{Prop_LongTermVol_4}
 \left\Vert H_{s}^{T} \right\Vert & \overset{(\ref{Prop_LongTermVol_3})}{=} \frac{2}{T} \left\Vert Q\, \Sigma\!\left(s,T\right) \sqrt{X_{s}}\,\right\Vert \nonumber\\
&  \overset{\phantom{(\ref{Prop_LongTermVol_3})}}{=} \frac{2}{T} \left(\operatorname{Tr}\!\left[\sqrt{X_{s}}\,\Sigma\!\left(s,T\right)Q^{\top}Q\,\Sigma\!\left(s,T\right)\sqrt{X_{s}}\,\right]\right)^{\nicefrac{1}{2}} \nonumber\\
&  \overset{\phantom{(\ref{Prop_LongTermVol_3})}}{=} \frac{2}{T} \left(\sum_{i,j,k,l,m,n} \sqrt{X_{ij,s}}\, \Sigma\!\left(s,T\right)_{jk} Q^{\top}_{kl} \, Q_{lm}\, \Sigma\!\left(s,T\right)_{mn} \sqrt{X_{ni,s}}\right)^{\nicefrac{1}{2}} \nonumber\\
&  \underset{\phantom{(\ref{Prop_LongTermVol_3})}}{\overset{(\ref{Assumption2_w})}{\leq}} \frac{2}{T} \left(\sum_{i,j,k,l,m,n} \sqrt{T} \, \sqrt{X_{ij,s}}\, w_{jk}\!\left(s\right) \, Q^{\top}_{kl} \, Q_{lm}\, \sqrt{T} \, w_{mn}\!\left(s\right) \, \sqrt{X_{ni,s}}\right)^{\nicefrac{1}{2}} \nonumber\\
&  \overset{\phantom{(\ref{Prop_LongTermVol_3})}}{=} \frac{2}{\sqrt{T}} \left(\sum_{i,j,k,l,m,n} \sqrt{X_{ij,s}}\, w_{jk}\!\left(s\right) \, Q^{\top}_{kl} \, Q_{lm}\, w_{mn}\!\left(s\right) \, \sqrt{X_{ni,s}}\right)^{\nicefrac{1}{2}} \nonumber\\
&  \overset{\phantom{(\ref{Prop_LongTermVol_3})}}{\leq} 2 \left(\sum_{i,j,k,l,m,n} \sqrt{X_{ij,s}}\, w_{jk}\!\left(s\right) \, Q^{\top}_{kl} \, Q_{lm}\, w_{mn}\!\left(s\right) \, \sqrt{X_{ni,s}}\right)^{\nicefrac{1}{2}} \nonumber\\
&  \overset{\phantom{(\ref{Prop_LongTermVol_3})}}{=} 2 \left(\operatorname{Tr}\!\left[\sqrt{X_{s}} \, w\!\left(s\right) \, Q^{\top}  Q \, w\!\left(s\right)  \sqrt{X_{s}}\right]\right)^{\nicefrac{1}{2}} \nonumber\\
&  \overset{\phantom{(\ref{Prop_LongTermVol_3})}}{=} 2 \left\Vert Q \, w\!\left(s\right) \sqrt{X_{s}}\right\Vert \equalscolon h\!\left(s\right). \nonumber
\end{align}
Further, it is $w \in L\!\left(W\right)$ due to Assumption \ref{Assumption2} and 
we know from Theorem \ref{Affine_Process_Representation} that $\sqrt{X} \in L\!\left(W\right)$. By using Theorem 16 in Chapter IV, Section 2 of \cite{Book_Protter} it follows $h \in L\!\left(W\right)$.
Then, applying the dominated convergence theorem for semimartingales (cf.\,\,Theorem 32 in Chapter IV, Section 2 of \cite{Book_Protter}), 
we get:
\begin{equation}\label{Prop_LongTermVol_5}
 \int_{0}^{t}\! \operatorname{Tr}\!\left[2\, \frac{Q\, \Sigma\!\left(s,T\right) \sqrt{X_{s}}}{T} \, dW_{s}\right] \overset{T \rightarrow \infty}{\longrightarrow} 2 \int_{0}^{t} \!  \operatorname{Tr}\!\left[\sigma_{\infty}\!\left(s\right) \sqrt{X_{s}} \, dW_{s}\, Q\right] \text{ in ucp.}
\end{equation}
It follows due to Lemma 5.8 of \cite{Book_Georgii},
 \eqref{Prop_LongTermVol_2}, and \eqref{Prop_LongTermVol_5}:
\begin{displaymath}
  2\,\int_{0}^{t}\! \operatorname{Tr}\!\left[\frac{\Sigma\!\left(s,T\right) - \Sigma\!\left(s,t\right)}{T-t} \sqrt{X_{s}} dW_{s} Q\right] \overset{T \rightarrow \infty}{\longrightarrow} 2 \int_{0}^{t} \!  \operatorname{Tr}\!\left[\sigma_{\infty}\!\left(s\right) \sqrt{X_{s}}\, dW_{s}\, Q\right] \text{ in ucp}.
\end{displaymath}
\end{proof}

\begin{proposition}\label{LongTermDrift0}
Under \emph{Assumption \ref{Assumption1}} and \emph{\ref{Assumption2}}, it holds for all $t \geq 0$:
\begin{equation}\label{Prop_LongTermDrift_Equ}
 \lim\limits_{T \rightarrow \infty} 2\,\int\limits_{0}^{t}\! \operatorname{Tr}\!\left[Q\,\frac{\Gamma\!\left(s,T\right) - \Gamma\!\left(s,t\right)}{T-t}\, Q^{\top}\right] ds = 2 \int\limits_{0}^{t} \!  \operatorname{Tr}\!\left[Q\, \mu_{\infty}\!\left(s\right) Q^{\top}\right] ds,
\end{equation}
where $\mu_{\infty}\!\left(s\right), s \geq 0$, is the long-term drift process defined by equation \eqref{LongTermDriftEqu}, 
$\Gamma\!\left(s,t\right)$, $s \geq 0$, is defined for all $t \geq 0$ as in \eqref{Definition_Gamma}, and the convergence in 
\eqref{Prop_LongTermDrift_Equ} is in ucp.
\end{proposition}

\begin{proof}
Fix $t \geq 0$. 
Since the process $\Gamma\!\left(s,t\right), s,t \geq 0,$ is continuous in $t$ for all fixed $s$ and c\`{a}dl\`{a}g in $s$ for all fixed $t$ 
it follows  by (4) of Section 2.8 in \cite{book_Applebaum} that $\Gamma\!\left(s,t\right), s\geq 0,$ 
is bounded on all compact intervals $\left[a,b\right]$ with $t \in \left[a,b\right]$ and $0 \leq a < b$ for a.e. $\omega \in \Omega$, i.e. 
\begin{equation}\label{Prop_LongTermDrift_00}
 \sup_{t \in \left[a,b\right]}\left\vert \int_{0}^{t}\! \operatorname{Tr}\!\left[Q\, \Gamma\!\left(s,t\right) Q^{\top}\right] ds \right\vert < \infty\ \,\mathbb{P}\text{-a.s.}\nonumber
\end{equation}
Consequently on every compact interval $\left[a,b\right]$
\begin{equation}\label{Prop_LongTermDrift_000}
 \frac{1}{T}\sup_{t \in \left[a,b\right]} \int_{0}^{t}\! \operatorname{Tr}\!\left[Q\, \Gamma\!\left(s,t\right) Q^{\top}\right] ds \overset{T \rightarrow \infty}{\longrightarrow} 0 \ \,\mathbb{P}\text{-a.s.}\nonumber
\end{equation}
Therefore
\begin{equation}\label{Prop_LongTermDrift_1}
\frac{1}{T} \int_{0}^{t}\! \operatorname{Tr}\!\left[Q\,\Gamma\!\left(s,t\right)Q^{\top}\right] ds  \overset{T \rightarrow \infty}{\longrightarrow} 0\, \text{ in ucp.}
\end{equation}
Let us define $G^{T} \colonequals G^{T}_{s}, s \geq 0$, with 
\begin{equation}\label{Prop_LongTermDrift_2}
 G^{T}_{s} \colonequals 2 \frac{Q\, \Gamma\!\left(s,T\right) Q^{\top}}{T}\,.
\end{equation}
Then for $T \rightarrow \infty:\ G_{s}^{T} \rightarrow 2\, Q \,\mu_{\infty}\!\left(s\right) Q^{\top}\ \text{a.s. for all}\, s \geq 0$.
\par
\noindent
By Assumption \ref{Assumption2} we have that for all $i,j \in \left\{1,\dots,d\right\}$ and $0 \leq s \leq T$:
\begin{align}\label{Prop_LongTermDrift_2_1}
 \Gamma\!\left(s,T\right)_{ij} & \overset{(\ref{Definition_Gamma})}{=} \left(\Sigma\!\left(s,T\right) X_{s} \Sigma\!\left(s,T\right)\right)_{ij}
 \overset{\phantom{(\ref{Definition_Gamma})}}{=} \sum_{k,l} \Sigma\!\left(s,T\right)_{ik} X_{kl,s} \Sigma\!\left(s,T\right)_{lj} \nonumber\\
& \underset{\phantom{(\ref{Definition_Gamma})}}{\overset{(\ref{Assumption2_w})}{\leq}} \sum_{k,l} \sqrt{T} \, w_{ik}\!\left(s\right) X_{kl,s}\, \sqrt{T} \, w_{lj}\!\left(s\right) 
 \overset{\phantom{(\ref{Definition_Gamma})}}{=} T \left(w\!\left(s\right) X_{s} w\!\left(s\right)\right)_{ij}\,.
\end{align}
Therefore we have that for all $0 \leq s \leq T$
\begin{align}\label{Prop_LongTermDrift_3}
 \left\Vert G_{s}^{T} \right\Vert & \overset{(\ref{Prop_LongTermDrift_2})}{=} \frac{2}{T} \left\Vert Q\, \Gamma\!\left(s,T\right) Q^{\top}\,\right\Vert \nonumber\\
& \overset{\phantom{(\ref{Definition_Gamma})}}{=} \frac{2}{T} \left(\operatorname{Tr}\!\left[Q \,\Gamma\!\left(s,T\right) Q^{\top}  Q \, \Gamma\!\left(s,T\right) Q^{\top}\right]\right)^{\nicefrac{1}{2}} \nonumber\\
& \overset{\phantom{(\ref{Definition_Gamma})}}{=} \frac{2}{T} \left(\sum_{i,j,k,l,m,n} Q_{ij} \, \Gamma\!\left(s,T\right)_{jk} Q^{\top}_{kl} \, Q_{lm} \, \Gamma\!\left(s,T\right)_{mn} Q^{\top}_{ni}\right)^{\nicefrac{1}{2}} \nonumber\\
& \underset{\phantom{(\ref{Definition_Gamma})}}{\overset{(\ref{Prop_LongTermDrift_2_1})}{\leq}} \frac{2}{T} \left(\sum_{i,j,k,l,m,n} Q_{ij} \, T \left(w\!\left(s\right)X_{s} w\!\left(s\right)\right)_{jk} Q^{\top}_{kl} \, Q_{lm} \, T \left(w\!\left(s\right)X_{s} w\!\left(s\right)\right)_{mn} Q^{\top}_{ni}\right)^{\nicefrac{1}{2}} \nonumber\\
& \overset{\phantom{(\ref{Definition_Gamma})}}{=} 2 \left(\operatorname{Tr}\!\left[Q \, w\!\left(s\right) X_{s} w\!\left(s\right) Q^{\top} Q \,  w\!\left(s\right) X_{s} w\!\left(s\right) Q^{\top} \right]\right)^{\nicefrac{1}{2}} \nonumber\\
& \overset{\phantom{(\ref{Definition_Gamma})}}{=} 2 \left\Vert Q \, w\!\left(s\right) X_{s} w\!\left(s\right) Q^{\top} \right\Vert \equalscolon g\!\left(s\right), \nonumber
\end{align}
where $g$ is a c\`{a}dl\`{a}g process. 
It follows $\int_{0}^{t} g\!\left(s\right)\,ds < \infty$ for all $t \geq 0$ by (4) of Section 2.8 in \cite{book_Applebaum} 
and we can apply the DCT for progressive processes (cf.\,\,Corollary 6.26 in Chapter 6 of \cite{Book_Klenke}). 
By using Lemma 5.8 of \cite{Book_Georgii} and \eqref{Prop_LongTermDrift_1} we obtain \eqref{Prop_LongTermDrift_Equ}.
\end{proof}

\begin{proposition}\label{LongTermJumps_Proposition}
 Under \emph{Assumption \ref{Assumption1}} and \emph{\ref{Assumption2}}, it holds for all $t \geq 0$:
\begin{equation}\label{Prop_LongTermJumps_Equ}
 \int_{0}^{t}\!\int_{S_{d}^{+}\! \setminus \left\{ 0\right\}}\! \frac{\operatorname{Tr}\!\left[\left(\Sigma\!\left(s,T\right) \!-\! \Sigma\!\left(s,t\right)\right)\xi\right]}{T-t}\, \mu^{X\!}\!\left(ds,d\xi\right) \!\overset{T \rightarrow \infty}{\longrightarrow}\!
\int_{0}^{t}\!\int_{S_{d}^{+}\! \setminus \left\{ 0\right\}}\!  \operatorname{Tr}\!\left[\sigma_{\infty}\!\left(s\right) \xi\right] \mu^{X\!}\!\left(ds,d\xi\right),
\end{equation}
where $\sigma_{\infty}\!\left(s\right), s \geq 0$, is the long-term volatility process defined by equation \eqref{LongTermVolaEqu}, 
$\Sigma\!\left(s,t\right)$, $s \geq 0$, is defined for all $t \geq 0$ as in \eqref{Sigma_Definition}, and the convergence in \eqref{Prop_LongTermJumps_Equ} is in ucp.
\end{proposition}

\begin{proof}
 Fix $t \geq 0$. 
First, notice that
\begin{align}\label{Prop_LongTermJumps_0}
 \left\vert \int_{0}^{t}\!\int_{S_{d}^{+}\! \setminus \left\{ 0\right\}}\!\operatorname{Tr}\!\left[\Sigma\!\left(s,t\right)\xi\right]\mu^{X\!}\!\left(ds,d\xi\right)\right\vert & \overset{\phantom{\eqref{Assumption2_w}}}{\leq} \int_{0}^{t}\!\int_{S_{d}^{+}}\!\left\vert \operatorname{Tr}\!\left[\Sigma\!\left(s,t\right)\xi\right]\right\vert\mu^{X\!}\!\left(ds,d\xi\right)\nonumber\\
& \overset{\phantom{\eqref{Assumption2_w}}}{\leq} \sqrt{t} \int_{0}^{t}\!\int_{S_{d}^{+}\! \setminus \left\{ 0\right\}}\!\frac{1}{\sqrt{t}}\left\Vert\Sigma\!\left(s,t\right)\right\Vert \left\Vert\xi\right\Vert \mu^{X\!}\!\left(ds,d\xi\right)\nonumber\\
& \overset{\eqref{Assumption2_w}}{\leq} \sqrt{t} \int_{0}^{t}\!\int_{S_{d}^{+}\! \setminus \left\{ 0\right\}}\! \left\Vert w\!\left(s\right)\right\Vert \left\Vert\xi\right\Vert \mu^{X\!}\!\left(ds,d\xi\right)\nonumber\\
& \overset{\phantom{\eqref{Assumption2_w}}}{\leq} \sqrt{t} \sup_{u \in\left[0,t\right]} \left\Vert w\!\left(u\right)\right\Vert \int_{0}^{t}\!\int_{S_{d}^{+}\! \setminus \left\{ 0\right\}}\! \left\Vert\xi\right\Vert \, \mu^{X\!}\!\left(ds,d\xi\right)\,.\nonumber
\end{align}
Define 
\begin{math}
 q\!\left(t\right) \colonequals \sqrt{t}\, \sup_{u \in\left[0,t\right]} \left\Vert w\!\left(u\right)\right\Vert Z_{t}
\end{math}
with 
\begin{displaymath}
 Z_{t} \colonequals \int_{0}^{t}\!\int_{S_{d}^{+}\! \setminus \left\{ 0\right\}}\! \left\Vert\xi\right\Vert \, \mu^{X\!}\!\left(ds,d\xi\right), t\geq 0\,.
\end{displaymath}
Note that $Z_{t}, t\geq 0,$ is a well-defined c\`{a}dl\`{a}g process by \eqref{Finite_Variation_Jumps_Remark_Equation}.
Then for all compact intervals $\left[a,b\right]$ with $0 \leq a < b$ it is 
\begin{align*}
 \sup_{t \in \left[a,b\right]}\left\vert q\!\left(t\right) \right\vert & = 
\sup_{t \in \left[a,b\right]}\left\vert \sqrt{t} \sup_{u \in \left[0,t\right]} \left\Vert w\!\left(u\right)\right\Vert Z_{t} \right\vert 
 \leq \sqrt{b} \sup_{t \in \left[a,b\right]} \sup_{u \in \left[0,t\right]}  \left\Vert w\!\left(u\right)\right\Vert \sup_{t \in \left[a,b\right]} Z_{t}\\
& = \sqrt{b} \sup_{u \in \left[0,b\right]}  \left\Vert w\!\left(u\right)\right\Vert \sup_{t \in \left[a,b\right]} Z_{t}
 < \infty
\end{align*}
because of (4) of Section 2.8 in \cite{book_Applebaum} applied for the c\`{a}dl\`{a}g processes $\left\Vert w\!\left(t\right) \right\Vert,t \geq 0,$ and $Z_{t}, t\geq 0$. 
Consequently on every compact interval $\left[a,b\right]$
\begin{equation}\label{Prop_LongTermJumps_1}
\frac{1}{T} \sup_{t \in \left[a,b\right]} \int_{0}^{t}\!\int_{S_{d}^{+}\! \setminus \left\{ 0\right\}}\! \operatorname{Tr}\!\left[\Sigma\!\left(s,t\right)\xi\right] \mu^{X\!}\!\left(ds,d\xi\right) \overset{T \rightarrow \infty}{\longrightarrow} 0\ \,\mathbb{P}\text{-a.s.}\nonumber
\end{equation}
Therefore
\begin{equation}\label{Prop_LongTermJumps_2}
\frac{1}{T} \int_{0}^{t}\!\int_{S_{d}^{+}\! \setminus \left\{ 0\right\}}\! \operatorname{Tr}\!\left[\Sigma\!\left(s,t\right)\xi\right] \mu^{X\!}\!\left(ds,d\xi\right) \overset{T \rightarrow \infty}{\longrightarrow} 0\, \text{ in ucp.}
\end{equation}
Due to Assumption~\ref{Assumption2} we have for $s,t \leq T$ and $T \geq 1$ that
\begin{equation}\label{Prop_LongTermJumps_3}
  \frac{\left\vert \operatorname{Tr}\left[\Sigma\!\left(s,T\right) \xi\right]\right\vert}{T} \leq \frac{1}{T} \left\Vert \Sigma\!\left(s,T\right)\right\Vert \left\Vert \xi\right\Vert
 \overset{\eqref{Assumption2_w}}{\leq} \frac{1}{\sqrt{T}} \left\Vert w\!\left(s\right)\right\Vert \left\Vert \xi\right\Vert 
\leq \left\Vert w\!\left(s\right)\right\Vert \left\Vert \xi\right\Vert \equalscolon j\!\left(s,\xi\right). \nonumber
\end{equation}
We first show that the process $j$ is integrable with respect to the random measure $\mu^{X\!}$ on $\left[0,t\right] \times {S_{d}^{+}\! \setminus\! \left\{ 0\right\}}$ 
for all $t\geq 0$:
\begin{align*}
\int_{0}^{t}\!\int_{S_{d}^{+}\! \setminus \left\{ 0\right\}}\! j\!\left(s,\xi\right) \mu^{X\!}\!\left(ds,d\xi\right)  & = 
\int_{0}^{t}\!\int_{S_{d}^{+}\! \setminus \left\{ 0\right\}}\!\left\Vert w\!\left(s\right)\right\Vert \left\Vert\xi\right\Vert \, \mu^{X\!}\!\left(ds,d\xi\right)\nonumber\\
  & \leq  \sup_{u \in \left[0,t\right]}  \left\Vert w\!\left(u\right)\right\Vert Z_{t} < \infty
\end{align*}
due to (4) of Section 2.8 in \cite{book_Applebaum} applied for the c\`{a}dl\`{a}g process $\left\Vert w\!\left(t\right) \right\Vert,t \geq 0,$ and by \eqref{Finite_Variation_Jumps_Remark_Equation}.
We have with the DCT and \eqref{LongTermVolaEqu} that for all fixed $t\geq 0$
\small
\begin{align}\label{Prop_LongTermJumps_5}
\int_{0}^{t}\!\int_{S_{d}^{+}\! \setminus \left\{ 0\right\}\!}\! \frac{\operatorname{Tr}\!\left[\Sigma\!\left(s,T\right)\xi\right]}{T} \mu^{X\!}\!\left(ds,d\xi\right) 
& \overset{T \rightarrow \infty}{\longrightarrow} \int_{0}^{t}\!\int_{S_{d}^{+}\! \setminus \left\{ 0\right\}\!}\!  \operatorname{Tr}\!\left[\sigma_{\infty}\!\left(s\right) \xi\right] \mu^{X\!}\!\left(ds,d\xi\right)\ \,\mathbb{P}\text{-a.s.}
\end{align}
\normalsize
Then, by \eqref{Prop_LongTermJumps_5} applied for $t=b$ and by \eqref{Lemma_Help_Proposition_Equation_0} from Lemma \ref{Lemma_Help_Proposition}, it follows that 
\begin{equation}\label{Prop_LongTermJumps_7}
 \sup_{t \in \left[a,b\right]} \int_{0}^{t}\!\int_{S_{d}^{+}\! \setminus \left\{ 0\right\}}\! \frac{\operatorname{Tr}\!\left[\Sigma\!\left(s,T\right)\xi\right]}{T} \mu^{X\!}\!\left(ds,d\xi\right)  \overset{T \rightarrow \infty}{\longrightarrow} \sup_{t \in \left[a,b\right]} \int_{0}^{t}\!\int_{S_{d}^{+}\! \setminus \left\{ 0\right\}}\!  \operatorname{Tr}\!\left[\sigma_{\infty\!}\!\left(s\right) \xi\right] \mu^{X\!}\!\left(ds,d\xi\right)\nonumber
\end{equation}
$\mathbb{P}$-a.s. and therefore in probability.
Hence by \cite{Book_Protter}, page 57
\begin{equation}\label{Prop_LongTermJumps_8}
 \lim\limits_{T \rightarrow \infty}\int_{0}^{t}\!\int_{S_{d}^{+}\! \setminus \left\{ 0\right\}\!}\! \frac{\operatorname{Tr}\left[\Sigma\!\left(s,T\right)\xi\right]}{T}\, \mu^{X\!}\!\left(ds,d\xi\right) = \int_{0}^{t}\!\int_{S_{d}^{+}\! \setminus \left\{ 0\right\}\!}\!  \operatorname{Tr}\!\left[\sigma_{\infty\!}\!\left(s\right) \xi\right] \mu^{X\!}\!\left(ds,d\xi\right) \, \text{ in ucp.}
\end{equation}
By using \eqref{Prop_LongTermJumps_2} as well as \eqref{Prop_LongTermJumps_8} it follows \eqref{Prop_LongTermJumps_Equ}.
\end{proof}

\begin{lemma}\label{Lemma_Help_Proposition}
 With $\Sigma\!\left(t,T\right), t\geq 0,$ defined as in \eqref{Sigma_Definition} and $\sigma_{\infty\!}\!\left(t\right), t\geq 0,$ as in 
\eqref{LongTermVolaEqu} it holds for $0 \leq a < b$
\small
\begin{align}\label{Lemma_Help_Proposition_Equation_0}
 & \left\vert \sup_{t \in \left[a,b\right]} \int_{0}^{t}\!\int_{S_{d}^{+}\! \setminus \left\{ 0\right\}\!}\! \frac{\operatorname{Tr}\!\left[\Sigma\!\left(s,T\right)\xi\right]}{T} \mu^{X\!}\!\left(ds,d\xi\right) - \sup_{t \in \left[a,b\right]} \int_{0}^{t}\!\int_{S_{d}^{+}\! \setminus \left\{ 0\right\}\!}\! \operatorname{Tr}\!\left[\sigma_{\infty\!}\!\left(s\right) \xi\right] \mu^{X\!}\!\left(ds,d\xi\right) \right\vert \phantom{====}\nonumber\\
 & \phantom{====} \leq \int_{0}^{b}\!\int_{S_{d}^{+}\! \setminus \left\{ 0\right\}}\! \left\vert \frac{\operatorname{Tr}\!\left[\Sigma\!\left(s,T\right)\xi\right]}{T} - \operatorname{Tr}\!\left[\sigma_{\infty\!}\!\left(s\right) \xi\right] \right\vert \mu^{X\!}\!\left(ds,d\xi\right).
\end{align}
\normalsize
\end{lemma}

\begin{proof}
 Let $0 \leq a \leq b$. Then, we have
\begin{align}\label{Lemma_Help_Proposition_Equation_1}
& \sup_{t \in \left[a,b\right]} \int_{0}^{t}\!\int_{S_{d}^{+}\! \setminus \left\{ 0\right\}\!}\! \frac{\operatorname{Tr}\!\left[\Sigma\!\left(s,T\right)\xi\right]}{T} \mu^{X\!}\!\left(ds,d\xi\right) - \sup_{t \in \left[a,b\right]} \int_{0}^{t}\!\int_{S_{d}^{+}\! \setminus \left\{ 0\right\}\!}\! \operatorname{Tr}\!\left[\sigma_{\infty\!}\!\left(s\right) \xi\right] \mu^{X\!}\!\left(ds,d\xi\right)\phantom{====}\nonumber\\
& \phantom{====} = \sup_{t \in \left[a,b\right]} \int_{0}^{t}\!\int_{S_{d}^{+}\! \setminus \left\{ 0\right\}\!}\!\left( \frac{\operatorname{Tr}\!\left[\Sigma\!\left(s,T\right)\xi\right]}{T} - \operatorname{Tr}\!\left[\sigma_{\infty\!}\!\left(s\right) \xi\right] + \operatorname{Tr}\!\left[\sigma_{\infty\!}\!\left(s\right) \xi\right]\right) \mu^{X\!}\!\left(ds,d\xi\right) \nonumber\\
& \phantom{=======} - \sup_{t \in \left[a,b\right]} \int_{0}^{t}\!\int_{S_{d}^{+}\! \setminus \left\{ 0\right\}\!}\! \operatorname{Tr}\!\left[\sigma_{\infty\!}\!\left(s\right) \xi\right] \mu^{X\!}\!\left(ds,d\xi\right) \nonumber\\
& \phantom{====} \leq \sup_{t \in \left[a,b\right]} \int_{0}^{t}\!\int_{S_{d}^{+}\! \setminus \left\{ 0\right\}\!}\!\left(\frac{\operatorname{Tr}\!\left[\Sigma\!\left(s,T\right)\xi\right]}{T} - \operatorname{Tr}\!\left[\sigma_{\infty\!}\!\left(s\right) \xi\right] \right)\mu^{X\!}\!\left(ds,d\xi\right).
\end{align}
Furthermore
\begin{align}\label{Lemma_Help_Proposition_Equation_2}
 & \sup_{t \in \left[a,b\right]} \int_{0}^{t}\!\int_{S_{d}^{+}\! \setminus \left\{ 0\right\}\!}\! \frac{\operatorname{Tr}\!\left[\Sigma\!\left(s,T\right)\xi\right]}{T} \mu^{X\!}\!\left(ds,d\xi\right) - \sup_{t \in \left[a,b\right]} \int_{0}^{t}\!\int_{S_{d}^{+}\! \setminus \left\{ 0\right\}\!}\! \operatorname{Tr}\!\left[\sigma_{\infty\!}\!\left(s\right) \xi\right] \mu^{X\!}\!\left(ds,d\xi\right)\phantom{====}\nonumber\\
& \phantom{====} = \sup_{t \in \left[a,b\right]} \int_{0}^{t}\!\int_{S_{d}^{+}\! \setminus \left\{ 0\right\}\!}\! \frac{\operatorname{Tr}\!\left[\Sigma\!\left(s,T\right)\xi\right]}{T} \mu^{X\!}\!\left(ds,d\xi\right) \nonumber\\
& \phantom{=======} - \sup_{t \in \left[a,b\right]} \int_{0}^{t}\!\int_{S_{d}^{+}\! \setminus \left\{ 0\right\}\!}\!\left( \operatorname{Tr}\!\left[\sigma_{\infty\!}\!\left(s\right) \xi\right] - \frac{\operatorname{Tr}\!\left[\Sigma\!\left(s,T\right)\xi\right]}{T} + \frac{\operatorname{Tr}\!\left[\Sigma\!\left(s,T\right)\xi\right]}{T}\right) \mu^{X\!}\!\left(ds,d\xi\right)\nonumber\\
& \phantom{====} \geq - \sup_{t \in \left[a,b\right]} \int_{0}^{t}\!\int_{S_{d}^{+}\! \setminus \left\{ 0\right\}\!}\!\left(\operatorname{Tr}\!\left[\sigma_{\infty\!}\!\left(s\right) \xi\right] - \frac{\operatorname{Tr}\!\left[\Sigma\!\left(s,T\right)\xi\right]}{T}\right) \mu^{X\!}\!\left(ds,d\xi\right)\nonumber\\
& \phantom{====} = \inf_{t \in \left[a,b\right]} \int_{0}^{t}\!\int_{S_{d}^{+}\! \setminus \left\{ 0\right\}\!}\!\left( \frac{\operatorname{Tr}\!\left[\Sigma\!\left(s,T\right)\xi\right]}{T} - \operatorname{Tr}\!\left[\sigma_{\infty\!}\!\left(s\right) \xi\right]\right) \mu^{X\!}\!\left(ds,d\xi\right).
\end{align}
Hence, it follows from \eqref{Lemma_Help_Proposition_Equation_1} and \eqref{Lemma_Help_Proposition_Equation_2} that
\begin{align}\label{Lemma_Help_Proposition_Equation_3}
 & \left\vert \sup_{t \in \left[a,b\right]} \int_{0}^{t}\!\int_{S_{d}^{+}\! \setminus \left\{ 0\right\}\!}\! \frac{\operatorname{Tr}\!\left[\Sigma\!\left(s,T\right)\xi\right]}{T} \mu^{X\!}\!\left(ds,d\xi\right) - \sup_{t \in \left[a,b\right]} \int_{0}^{t}\!\int_{S_{d}^{+}\! \setminus \left\{ 0\right\}\!}\! \operatorname{Tr}\!\left[\sigma_{\infty\!}\!\left(s\right) \xi\right] \mu^{X\!}\!\left(ds,d\xi\right) \right\vert \phantom{====}\nonumber\\
 & \phantom{====} \leq \left\vert \sup_{t \in \left[a,b\right]}  \int_{0}^{t}\!\int_{S_{d}^{+}\! \setminus \left\{ 0\right\}\!}\!\left( \frac{\operatorname{Tr}\!\left[\Sigma\!\left(s,T\right)\xi\right]}{T} - \operatorname{Tr}\!\left[\sigma_{\infty\!}\!\left(s\right) \xi\right]\right) \mu^{X\!}\!\left(ds,d\xi\right) \right\vert \nonumber\\
 & \phantom{=======} \vee \left\vert \inf_{t \in \left[a,b\right]}  \int_{0}^{t}\!\int_{S_{d}^{+}\! \setminus \left\{ 0\right\}\!}\!\left( \frac{\operatorname{Tr}\!\left[\Sigma\!\left(s,T\right)\xi\right]}{T} - \operatorname{Tr}\!\left[\sigma_{\infty\!}\!\left(s\right) \xi\right]\right) \mu^{X\!}\!\left(ds,d\xi\right) \right\vert \nonumber\\
 & \phantom{====} \leq \sup_{t \in \left[a,b\right]}  \int_{0}^{t}\!\int_{S_{d}^{+}\! \setminus \left\{ 0\right\}}\!\left\vert\frac{\operatorname{Tr}\!\left[\Sigma\!\left(s,T\right)\xi\right]}{T} - \operatorname{Tr}\!\left[\sigma_{\infty\!}\!\left(s\right) \xi\right]\right\vert \mu^{X\!}\!\left(ds,d\xi\right). \nonumber
\end{align}

\end{proof}

\begin{proposition}\label{LongTermJumps_Vanish_Proposition}
 Under \emph{Assumption \ref{Assumption1}} and \emph{\ref{Assumption2}}, it holds for all $t \geq 0$:
\begin{equation}\label{Prop_LongTermJumps_Vanish_Equ}
 \lim\limits_{T \rightarrow \infty} \int_{0}^{t}\!\int_{S_{d}^{+}\! \setminus \left\{ 0\right\}}\! \frac{e^{\operatorname{Tr}\left[\Sigma\left(s,T\right)\,\xi\right]} - e^{\operatorname{Tr}\left[\Sigma\left(s,t\right)\,\xi\right]}}{T - t}\, \nu\!\left(ds,d\xi\right) = 
0\,,
\end{equation}
where $\Gamma\!\left(s,t\right)$, $s \geq 0$, is defined for all $t \geq 0$ as in \eqref{Definition_Gamma}, and the convergence in \eqref{Prop_LongTermJumps_Vanish_Equ} is in ucp.
\end{proposition}

\begin{proof}
Fix $t\geq 0$. We note that the left-hand side of \eqref{Prop_LongTermJumps_Vanish_Equ} is equal to
\small
\begin{equation}\label{Prop_LongTermJumps_Vanish_0}
 \lim\limits_{T \rightarrow \infty} \frac{1}{T}\!\left(\int_{0}^{t}\!\int_{S_{d}^{+}\! \setminus \left\{ 0\right\}\!}\! \left(1\!-\!e^{\operatorname{Tr}\left[\Sigma\left(s,t\right)\,\xi\right]}\right) \nu\!\left(ds,d\xi\right) - \int_{0}^{t}\!\int_{S_{d}^{+}\! \setminus \left\{ 0\right\}\!}\! \left(1\!-\!e^{\operatorname{Tr}\left[\Sigma\left(s,T\right)\,\xi\right]}\right) \nu\!\left(ds,d\xi\right)\!\right).
\end{equation}
\normalsize
Hence we study the limit \eqref{Prop_LongTermJumps_Vanish_0} and
we introduce for all $u \in S_{d}^{+}$
 \begin{align}
 & \tilde{F}\!\left(u\right) \colonequals \int_{S_{d}^{+}\! \setminus \left\{ 0\right\}\!}\!\left(e^{-\operatorname{Tr}\left[u \xi\right]} - 1\right)\, m\!\left(d\xi\right)\,, \label{LongTermJumps_Vanish_Help_2}\\
 & \tilde{R}\!\left(u\right) \colonequals \int_{S_{d}^{+}\! \setminus \left\{ 0\right\}\!}\!\left(e^{-\operatorname{Tr}\left[u \xi\right]} - 1\right)\, \mu\!\left(d\xi\right)\,. \label{LongTermJumps_Vanish_Help_3}
 \end{align}
Then, we can write due to \eqref{Affine_Process_Representation_Equation_2}, \eqref{LongTermJumps_Vanish_Help_2}, and \eqref{LongTermJumps_Vanish_Help_3} that
\begin{equation}\label{LongTermJumps_Vanish_Help_4}
\int_{0}^{t}\!\int_{S_{d}^{+}\! \setminus \left\{ 0\right\}\!}\! \left(1\!-\!e^{\operatorname{Tr}\left[\Sigma\left(s,t\right)\xi\right]}\right) \nu\!\left(ds,d\xi\right) = -\int_{0}^{t}\!\left(\tilde{F}\!\left(-\Sigma\!\left(s,t\right)\right) + \operatorname{Tr}\!\left[\tilde{R}\!\left(-\Sigma\!\left(s,t\right)\right)X_{s}\right]\right)ds.\nonumber
\end{equation}
The process $\tilde{F}\!\left(-\Sigma\!\left(s,t\right)\right) + \operatorname{Tr}\!\left[\tilde{R}\!\left(-\Sigma\!\left(s,t\right)\right)X_{s}\right], s \in \left[0,t\right],$ is c\`{a}dl\`{a}g for all $t\geq 0$ since $\tilde{F}$ and $\tilde{R}$ are continuous functions and $X$ is c\`{a}dl\`{a}g. 
Due to this and (4) of Section 2.8 in \cite{book_Applebaum} we get that for all compact intervals $\left[a,b\right]$ with $a,b \geq 0$
\begin{equation}\label{Prop_LongTermJumps__Vanish_1}
 \sup_{t \in \left[a,b\right]} \int_{0}^{t}\!\int_{S_{d}^{+}\! \setminus \left\{ 0\right\}\!}\!\left(1\!-\!e^{\operatorname{Tr}\left[\Sigma\left(s,t\right)\,\xi\right]}\right)\, \nu\!\left(ds,d\xi\right) < \infty\ \,\mathbb{P}\text{-a.s.}\nonumber
\end{equation}
Consequently on every compact interval $\left[a,b\right]$
\begin{equation}\label{Prop_LongTermJumps_Vanish_11}
\frac{1}{T} \sup_{t \in \left[a,b\right]} \int_{0}^{t}\!\int_{S_{d}^{+}\! \setminus \left\{ 0\right\}\!}\! \left(1\!-\!e^{\operatorname{Tr}\left[\Sigma\left(s,t\right)\,\xi\right]}\right)\, \nu\!\left(ds,d\xi\right) \overset{T \rightarrow \infty}{\longrightarrow} 0\ \,\mathbb{P}\text{-a.s.}\nonumber
\end{equation}
Therefore
\begin{equation}\label{Prop_LongTermJumps_Vanish_2}
\frac{1}{T} \int_{0}^{t}\!\int_{S_{d}^{+}\! \setminus \left\{ 0\right\}\!}\! \left(1\!-\!e^{\operatorname{Tr}\left[\Sigma\left(s,t\right)\,\xi\right]}\right)\, \nu\!\left(ds,d\xi\right) \overset{T \rightarrow \infty}{\longrightarrow} 0\, \text{ in ucp.}
\end{equation}
Next, we use the inequality
\begin{equation}\label{Prop_LongTermJumps_Inequ1}
 1\! -\! e^{\operatorname{Tr}\left[\Sigma\left(s,T\right)\,\xi\right]} =   1\! -\! e^{-\operatorname{Tr}\left[-\!\Sigma\left(s,T\right)\,\xi\right]} \leq 1\! \land\! \operatorname{Tr}\!\left[-\!\Sigma\!\left(s,T\right)\xi\right]
\overset{\eqref{Assumption2_w}}{\leq}  1\! \land\! \sqrt{T} \operatorname{Tr}\!\left[w\!\left(s\right)\xi\right]
\end{equation}
which holds for all $\xi \in S_{d}^{+}$ and a.e. $\omega \in \Omega$, to see that for all $0 \leq s \leq T$ with $T \geq 1$, and for a.e. $\omega \in \Omega$
\begin{align}\label{Prop_LongTermJumps_Inequ2}
 \frac{1-e^{\operatorname{Tr}\left[\Sigma\left(s,T\right)\,\xi\right]}}{T} & \overset{\eqref{Prop_LongTermJumps_Inequ1}}{\leq} \frac{1}{T} \land \frac{1}{\sqrt{T}} \operatorname{Tr}\!\left[w\!\left(s\right)\xi\right] 
\leq 1 \land \operatorname{Tr}\!\left[w\!\left(s\right)\xi\right] \leq  1 \land \left\Vert w\!\left(s\right)\right\Vert \left\Vert \xi \right\Vert \equalscolon i\!\left(s,\xi\right)\,.\nonumber
\end{align}
Since we investigate long-term interest rates it is sufficient to impose long times of maturity, say $T\geq 1$. 
\par
\noindent
Then, we show that the process $i$ is integrable with respect to the random measure $\nu$ on $\left[0,t\right] \times S_{d}^{+}$:
\small
\begin{align*}
 \int_{0}^{t}\!\int_{S_{d}^{+}\! \setminus \left\{ 0\right\}}\! i\!\left(s,\xi\right) \nu\!\left(ds,d\xi\right) 
& \overset{\phantom{\eqref{Affine_Process_Representation_Equation_2}}}{=} \int_{0}^{t}\!\int_{S_{d}^{+}\! \setminus \left\{ 0\right\}}\! i\!\left(s,\xi\right) \left( \mathbbm{1}_{\left\{\left\Vert w\!\left(s\right) \right\Vert \leq 1\right\}} +  \mathbbm{1}_{\left\{\left\Vert w\!\left(s\right) \right\Vert > 1\right\}}\right) \nu\!\left(ds,d\xi\right) \nonumber\\
& \overset{\phantom{\eqref{Affine_Process_Representation_Equation_2}}}{=} \int_{0}^{t}\!\int_{S_{d}^{+}\! \setminus \left\{ 0\right\}}\! i\!\left(s,\xi\right) \mathbbm{1}_{\left\{\left\Vert w\!\left(s\right) \right\Vert \leq 1\right\}}\, \nu\!\left(ds,d\xi\right)  \nonumber\\
& \phantom{===} + \int_{0}^{t}\!\int_{S_{d}^{+}\! \setminus \left\{ 0\right\}}\! i\!\left(s,\xi\right) \mathbbm{1}_{\left\{\left\Vert w\!\left(s\right)\right\Vert  > 1\right\}}\, \nu\!\left(ds,d\xi\right) \nonumber\\
& \overset{\phantom{\eqref{Affine_Process_Representation_Equation_2}}}{\leq} \int_{0}^{t}\!\int_{S_{d}^{+}\! \setminus \left\{ 0\right\}}\!\mathbbm{1}_{\left\{\left\Vert w\!\left(s\right) \right\Vert \leq 1\right\}} \left(1 \land \left\Vert \xi\right\Vert\right) \nu\!\left(ds,d\xi\right) \nonumber\\
& \phantom{===} + \int_{0}^{t}\!\int_{S_{d}^{+}\! \setminus \left\{ 0\right\}}\!\mathbbm{1}_{\left\{\left\Vert w\!\left(s\right)\right\Vert > 1\right\}} \left\Vert w\!\left(s\right)\right\Vert\left(1 \land \left\Vert \xi\right\Vert\right) \nu\!\left(ds,d\xi\right) \nonumber\\
& \overset{\eqref{Affine_Process_Representation_Equation_2}}{\leq} \int_{0}^{t}\!\int_{S_{d}^{+}\! \setminus \left\{ 0\right\}}\!\left(1 \land \left\Vert \xi\right\Vert\right) \nu\!\left(ds,d\xi\right) \nonumber\\
& \phantom{===} +  \int_{0}^{t}\!\left\Vert w\!\left(s\right)\right\Vert \mathbbm{1}_{\left\{\left\Vert w\!\left(s\right)\right\Vert > 1\right\}}\, ds \int_{S_{d}^{+}\! \setminus \left\{ 0\right\}}\!\left(1 \land \left\Vert\xi\right\Vert\right) m\!\left(d\xi\right) \nonumber\\
& \phantom{===} +  \operatorname{Tr}\!\left[\int_{0}^{t}\!\left\Vert w\!\left(s\right)\right\Vert \mathbbm{1}_{\left\{\left\Vert w\!\left(s\right)\right\Vert > 1\right\}} X_{s}\, ds \int_{S_{d}^{+}\! \setminus \left\{ 0\right\}}\!\left(1 \land \left\Vert\xi\right\Vert\right) \mu\!\left(d\xi\right)\right]\nonumber\\
& \overset{\phantom{\eqref{Affine_Process_Representation_Equation_2}}}{\leq} \int_{0}^{t}\!\int_{S_{d}^{+}\! \setminus \left\{ 0\right\}}\!\left(1 \land \left\Vert \xi\right\Vert\right) \nu\!\left(ds,d\xi\right) \nonumber\\
& \phantom{===} +  \int_{0}^{t}\!\left\Vert w\!\left(s\right)\right\Vert \, ds \int_{S_{d}^{+}\! \setminus \left\{ 0\right\}}\!\left(1 \land \left\Vert\xi\right\Vert\right) m\!\left(d\xi\right) \nonumber\\
& \phantom{===} +  \operatorname{Tr}\!\left[\int_{0}^{t}\!\left\Vert w\!\left(s\right)\right\Vert X_{s}\, ds \int_{S_{d}^{+}\! \setminus \left\{ 0\right\}}\!\left(1 \land \left\Vert\xi\right\Vert\right) \mu\!\left(d\xi\right)\right]\nonumber\\
& \overset{\phantom{\eqref{Affine_Process_Representation_Equation_2}}}{<} \infty\ \,\mathbb{P}\text{-a.s.}\nonumber
\end{align*}
\normalsize
because of \eqref{Constant_Jump_Term_Equation}, \eqref{Linear_Jump_Coefficient_Equation}, and (4) of Section 2.8 in \cite{book_Applebaum} applied for the c\`{a}dl\`{a}g processes $X$ and $\left\Vert w\!\left(t\right) \right\Vert X_{t}, t\geq 0$.
\par
\noindent
Then by the DCT we have that for all $t\geq 0$
\begin{equation}\label{Prop_LongTermJumps_Vanish_3}
 \int_{0}^{t}\!\int_{S_{d}^{+}\! \setminus \left\{ 0\right\}\!}\! \frac{e^{\operatorname{Tr}\left[\Sigma\left(s,T\right)\,\xi\right]}\!-\!1}{T}\, \nu\!\left(ds,d\xi\right) \overset{T \rightarrow \infty}{\longrightarrow} 0\ \,\mathbb{P}\text{-a.s.}
\end{equation}
With the same argument as in Proposition \ref{LongTermJumps_Proposition}, we then obtain that for $0\leq a < b$
\begin{equation}\label{Prop_LongTermJumps_Vanish_41}
 \sup_{t \in \left[a,b\right]} \int_{0}^{t}\!\int_{S_{d}^{+}\! \setminus \left\{ 0\right\}\!}\! \frac{1\!-\!e^{\operatorname{Tr}\left[\Sigma\left(s,T\right)\,\xi\right]}}{T}\, \nu\!\left(ds,d\xi\right)  
\leq  \int_{0}^{b}\!\int_{S_{d}^{+}\! \setminus \left\{ 0\right\}\!}\! \frac{1\!-\!e^{\operatorname{Tr}\left[\Sigma\left(s,T\right)\,\xi\right]}}{T}  \nu\!\left(ds,d\xi\right),\nonumber
\end{equation}
since $\nu$ is given by \eqref{Affine_Process_Representation_Equation_2}.
\par
\noindent
This converges to $0$ by \eqref{Prop_LongTermJumps_Vanish_3} applied for $t = b$.
Therefore a.e. $\omega \in \Omega$
\begin{equation}\label{Prop_LongTermJumps_Vanish_4}
 \sup_{t \in \left[a,b\right]} \int_{0}^{t}\!\int_{S_{d}^{+}\! \setminus \left\{ 0\right\}\!}\! \frac{1\!-\!e^{\operatorname{Tr}\left[\Sigma\left(s,T\right)\,\xi\right]}}{T}\, \nu\!\left(ds,d\xi\right) \overset{T \rightarrow \infty}{\longrightarrow} 0\,,\nonumber
\end{equation}
i.\,\,e.\,\,by \cite{Book_Protter}, page 57
\begin{equation}\label{Prop_LongTermJumps_Vanish_5}
 \int_{0}^{t}\!\int_{S_{d}^{+}\! \setminus \left\{ 0\right\}\!}\! \frac{e^{\operatorname{Tr}\left[\Sigma\left(s,T\right)\,\xi\right]}\!-\!1}{T}\, \nu\!\left(ds,d\xi\right) \overset{T \rightarrow \infty}{\longrightarrow} 0\, \text{ in ucp.}
\end{equation}
The result \eqref{Prop_LongTermJumps_Vanish_Equ} follows then by \eqref{Prop_LongTermJumps_Vanish_0}, \eqref{Prop_LongTermJumps_Vanish_2}, and \eqref{Prop_LongTermJumps_Vanish_5}.
\end{proof}

We are now ready to state the main result of this section.

\begin{theorem}\label{Main_Long_Term_Yield}
 Under  \emph{Assumption \ref{Assumption1}} and \emph{\ref{Assumption2}}, the long-term yield is given by
\begin{equation}\label{LongTermYield1}
 \ell_{t} = \ell_{0} + 2\, \int_{0}^{t}\! \operatorname{Tr}\!\left[Q\, \mu_{\infty}\!\left(s\right) Q^{\top}\right] ds,\, t\geq 0, 
\end{equation}
with $\operatorname{Tr}\!\left[Q\, \mu_{\infty}\!\left(s\right) Q^{\top}\right] \geq 0$ for all $0 \leq s \leq t$ if it exists in a finite form.
\end{theorem}

\begin{proof}
 By Lemma \ref{Cont_Comp_Spot_Rate_Lemma}, Proposition \ref{LongTermYield0}, Proposition \ref{LongTermVola0}, Proposition \ref{LongTermDrift0}, Proposition \ref{LongTermJumps_Proposition}, and Proposition \ref{LongTermJumps_Vanish_Proposition},  the long-term yield can be written in the following way:
\begin{align*}
   \ell_{t} & = \ell_{0} + 2 \int_{0}^{t}\! \operatorname{Tr}\!\left[Q\,\mu_{\infty}\!\left(s\right)Q^{\top}\right]\,ds - 2 \int_{0}^{t}\!\operatorname{Tr}\!\left[\sigma_{\infty}\!\left(s\right) \sqrt{X_{s}}\, dW_{s}\, Q\right] \nonumber\\
& \phantom{====} - \int_{0}^{t}\!\int_{S_{d}^{+}\! \setminus \left\{ 0\right\}\!}\!  \operatorname{Tr}\!\left[\sigma_{\infty}\!\left(s\right) \xi\right] \mu^{X\!}\!\left(ds,d\xi\right),\, t\geq 0,
\end{align*}
whereas the convergence is uniformly on compacts in probability. 
However if $0 < \left\Vert \sigma_{\infty}\!\left(t\right) \right\Vert < \infty$ for some $t \in \left[0,T\right]$, by \eqref{LongTermVolaEqu} we have 
$\Sigma\!\left(t,T\right)_{ij} \in \operatorname{O}\!\left(T-t\right)$ for all $i,j \in \left\{1,\dots,d\right\}$. 
Then we get for all $t\geq 0$
\begin{align*}
 \operatorname{Tr}\!\left[Q\, \mu_{\infty}\!\left(t\right) Q^{\top}\right] & \overset{\eqref{LongTermDriftEqu}}{=} \sum_{i,j,k} Q_{ij} \lim_{T \rightarrow \infty} \frac{\Gamma\!\left(t,T\right)_{jk}}{T-t}\, Q^{\top}_{ki}\nonumber\\
& \overset{\eqref{Definition_Gamma}}{=} \lim_{T \rightarrow \infty} \frac{1}{T-t} \sum_{i,j,k,l,m} Q_{ij}\, \Sigma\!\left(t,T\right)_{jl} X_{lm,t}\, \Sigma\!\left(t,T\right)_{mk} Q_{ik} \nonumber\\
& = \infty\ \,\mathbb{P}\text{-a.s.}
\end{align*}
That is a contradiction to the existence of the long-term yield. 
Hence representation \eqref{LongTermYield1} holds if the long-term yield exists in a finite form. 
Moreover, we have for all $t\geq 0$
\begin{align}\label{Non_Negative_Long_Term_Drift_Equation_1}
 \operatorname{Tr}\!\left[Q\, \mu_{\infty}\!\left(t\right) Q^{\top}\right] & \underset{\phantom{(\ref{Definition_Gamma})}}{\overset{(\ref{LongTermDriftEqu})}{=}} \lim\limits_{T \rightarrow \infty} \frac{1}{T-t} \operatorname{Tr}\!\left[Q \, \Gamma\!\left(t,T\right) Q^{\top}\right]\nonumber \\
& \overset{(\ref{Definition_Gamma})}{=} \lim\limits_{T \rightarrow \infty} \frac{1}{T-t} \operatorname{Tr}\!\left[Q \, \Sigma\!\left(t,T\right) X_{t}\, \Sigma\!\left(t,T\right) Q^{\top}\right] \nonumber\\
& \overset{\phantom{(\ref{Definition_Gamma})}}{=} \lim\limits_{T \rightarrow \infty} \frac{1}{T-t} \left\Vert \sqrt{X_{t}}\, \Sigma\!\left(t,T\right) Q^{\top}\right\Vert^{2} \geq 0\ \,\mathbb{P}\text{-a.s.}
\end{align}
This concludes the proof.
\end{proof}

It follows immediately that $\left(\ell_{t}\right)_{t \geq 0}$ is a non-decreasing process. 
This is a well-known result, which was shown for the first time in 1996 by Dybvig, Ingersoll and Ross in \cite{dybvig96} and generally proven in \cite{article_Hubalek}. 
Our main contribution is to compute explicitly the form of the long-term yield in dependence of the model's parameters.
In particular, we obtain that 
the drift $\mu_{\infty}$ is given by a stochastic process which depends on (the limit of) 
the volatility and on $X$.  
This implies that the representation \eqref{LongTermYield1} of $\ell$ remains the same under a change 
of equivalent probability measures. This result can be proven by applying the convergence results of 
Propositions \ref{LongTermJumps_Proposition} and \ref{LongTermJumps_Vanish_Proposition} to the yield expressed in the form \eqref{Remark_Yield_Equation_0}, 
which provides the form of $Y$ under a change of equivalent probability measures. 
In this way we extend a result of \cite{Karoui97} and \cite{bh2012} to a multifactor setting.
\par
To conclude we now discuss some conditions on the volatility process $\sigma\!\left(t,T\right)$ that guarantee the existence of the long-term drift $\mu_{\infty}$.

\begin{proposition}\label{Behaviour_Long_Term_Drift_Proposition_1}
Let $\sigma\!\left(t,T\right) \in \operatorname{O}\!\left(\frac{1}{\sqrt{T-t}}\right)$ for every $t \geq 0$, i.e. 
$\sigma\!\left(t,T\right)_{ij} \in \operatorname{O}\!\left(\frac{1}{\sqrt{T-t}}\right)$ for all $i,j \in \left\{1,\dots,d\right\}$ 
$\mathbb{P}$-a.s.
Under the setting outlined in \emph{Section \ref{Affine_HJM_Framework}}, we get
  \begin{equation}\label{Behaviour_Long_Term_Drift_Equation_0}
\operatorname{Tr}\!\left[Q\, \mu_{\infty}\!\left(t\right) Q^{\top}\right] < \infty\,\ \mathbb{P}\text{-a.s.}  \nonumber 
  \end{equation}
\end{proposition}

\begin{proof}
 Let $t \geq 0$ and $\sigma\!\left(t,T\right)  \in \operatorname{O}\!\left(\frac{1}{\sqrt{T-t}}\right)$. 
 Then 
\begin{align*}\label{Behaviour_Long_Term_Drift_Equation_1_1}
 \operatorname{Tr}\!\left[Q\, \mu_{\infty}\!\left(t\right) Q^{\top}\right] & \overset{(\ref{Non_Negative_Long_Term_Drift_Equation_1})}{=} \lim\limits_{T \rightarrow \infty} \frac{1}{T\!-\!t} \left\Vert \sqrt{X_{t}}\, \Sigma\!\left(t,T\right) Q^{\top}\right\Vert^{2} \nonumber\\
& \underset{\phantom{(\ref{Non_Negative_Long_Term_Drift_Equation_1})}}{\overset{(\ref{Sigma_Definition})}{=}} \lim\limits_{T \rightarrow \infty} \frac{1}{T\!-\!t} \sum_{i,j,k,l,m} Q_{ij}\, \int\limits_{t}^{T}\! \sigma\!\left(t,u\right)_{jk}\! du \  X_{kl,t}\, \int\limits_{t}^{T}\! \sigma\!\left(t,u\right)_{lm}\! du\  Q_{mi}^{\top}\nonumber\\
& \underset{\phantom{\eqref{Non_Negative_Long_Term_Drift_Equation_1}}}{<} \infty\ \,\mathbb{P}\text{-a.s.}
\end{align*}
\end{proof}

\begin{proposition}\label{Behaviour_Long_Term_Drift_Proposition_2}
  Let $\sigma\!\left(t,T\right) \in \operatorname{O}\!\left(\frac{1}{T-t}\right)$ for every $t \geq 0$ $\mathbb{P}$-a.s. Under the setting outlined in \emph{Section \ref{Affine_HJM_Framework}}, we get
\begin{equation}\label{Behaviour_Long_Term_Drift_Equation_2_1}
 \mu_{\infty}\!\left(t\right) = 0\nonumber
\end{equation}
 and therefore $\left(\ell_{t}\right)_{t \geq 0}$ is constant.
\end{proposition}

\begin{proof}
Let $t \geq 0$ and $\sigma\!\left(t,T\right) \in \operatorname{O}\!\left(\frac{1}{T-t}\right)$, i.e. for all $i,j \in \left\{1,\dots,d\right\}$ 
it is $\sigma\!\left(t,T\right)_{ij} \in \operatorname{O}\!\left(\frac{1}{T-t}\right)$. Then, we get for all $i,j \in \left\{1,\dots,d\right\}$ that 
$\Sigma\!\left(t,T\right)_{ij} \in \operatorname{O}\!\left(\log\!\left(T-t\right)\right)$ and therefore for all $i,j,k,l \in \left\{1,\dots,d\right\}$ that 
\begin{equation}\label{Behaviour_Long_Term_Drift_Equation_2_2}
 \lim\limits_{T \rightarrow \infty} \frac{1}{T-t} \, \Sigma\!\left(t,T\right)_{ij} \, \Sigma\!\left(t,T\right)_{kl} = 0\ \,\mathbb{P}\text{-a.s.}
\end{equation}
Hence, for all $i,j \in \left\{1,\dots,d\right\}$ it is
\begin{equation}
 \mu_{\infty\!}\!\left(t\right)_{ij} \overset{(\ref{LongTermDriftEqu})}{=} \lim\limits_{T \rightarrow \infty\!}\! \frac{\Gamma\!\left(t,T\right)_{ij}}{T-t} \overset{(\ref{Definition_Gamma})}{=}  \lim\limits_{T \rightarrow \infty}\! \frac{1}{T\!-\!t} \sum_{k,l} \Sigma\!\left(t,T\right)_{ik}  X_{kl,t}  \Sigma\!\left(t,T\right)_{lj} \overset{(\ref{Behaviour_Long_Term_Drift_Equation_2_2})}{=}\! 0\ \,\mathbb{P}\text{-a.s.}\nonumber
\end{equation}
By \eqref{LongTermYield1} this yields 
$\ell_{t} = \ell_{0}$ for all $t \geq 0$, i.e. $\left(\ell_{t}\right)_{t \geq 0}$ is constant.
\end{proof}

The following table summarises the results regarding the convergence behaviour of the long-term yield for all $t\geq 0$.

\par
\noindent
\newline
\begin{tabular}[c]{||p{3.45cm}|p{2.83cm}|p{2.24cm}|p{2.83cm}||}\hline
\small{Long-term drift} & \small{Long-term volatility} & \small{Long-term yield} & \small{Volatility curve}\\ \hline \hline
\small{$\operatorname{Tr}\!\left[Q\,\mu_{\infty}\!\left(t\right)Q^{\top}\right] = \infty$} & \small{$0 < \left\Vert \sigma_{\infty}\!\left(t\right)\right\Vert < \infty$} & \small{infinite} & \small{$\sigma\!\left(t,T\right) \sim \operatorname{O}\!\left(1\right)$}\phantom{\huge{I}}\\ \hline
\small{$\operatorname{Tr}\!\left[Q\,\mu_{\infty}\!\left(t\right)Q^{\top}\right] = \infty$} & \small{$0 < \left\Vert \sigma_{\infty}\!\left(t\right)\right\Vert < \infty$} & \small{infinite} & \small{$\sigma\!\left(t,T\right) \sim \operatorname{O}\!\left(T\!-\!t\right)$}\phantom{\huge{I}}\\ \hline
\small{$\operatorname{Tr}\!\left[Q\,\mu_{\infty}\!\left(t\right)Q^{\top}\right] = 0$} & \small{$\left\Vert \sigma_{\infty}\!\left(t\right)\right\Vert = 0$} & \small{constant} & \small{$\sigma\!\left(t,T\right) \sim \operatorname{O}\!\left(\frac{1}{T-t}\right)$}\phantom{\huge{I}}\\ \hline
\small{$0\! <\! \operatorname{Tr}\!\left[Q\, \mu_{\infty}\!\left(t\right)Q^{\top}\right]\! <\! \infty$} & \small{$\left\Vert \sigma_{\infty}\!\left(t\right)\right\Vert = 0$}\phantom{\huge{I}} & \small{non-decreasing} & \small{$\sigma\!\left(t,T\right) \sim \operatorname{O}\!\left(\frac{1}{\sqrt{T-t}}\right)$}\\ \hline
\end{tabular}

\section{Examples}\label{Examples}

In this section we present two examples for the long-term yield driven by an affine process on $S_{d}^{+}$ under two different volatility specifications. 
Detailed calculations of this example can be found in \cite{Haertel_Thesis}.
Consider the affine process on $S_{d}^{+}$ defined as 
\begin{equation}\label{Example_1}
 dX_{t} \colonequals a I_{d} dt + \sqrt{X_{t}} dW_{t} + dW_{t}^{\top} \sqrt{X_{t}}
\end{equation}
for $t \geq 0$, $a \geq \left(d - 1\right)$, $W$ a $d$-dimensional matrix Brownian motion, and $I_{d}$ denoting the $d$-dimensional unit matrix. 
Here we model directly under the equivalence martingale measure $\mathbb{Q}$, i.e.\,\,we assume\,\,$\nu = \nu^{\ast}$ which implies $K\!\left(s,\xi\right) = 1$ for all $s \geq 0$ and 
$\xi \in S_{d}^{+}$, and $W = W^{\ast}$ which implies $\gamma \equiv 0$.
\par
By Theorem \ref{HJM_Drift_Condition_Theorem} the $\mathbb{Q}$-dynamics of the forward rate are given by
\small
\begin{equation}
\label{Example_2}
f\!\left(t,T\right) = f\!\left(0,T\right) + 4 \int_{0}^{t}\! \operatorname{Tr}\left[\sigma\!\left(s,T\right) X_{s}\, \left(-\Sigma\!\left(s,T\right)\right)\right]ds + 2 \int_{0}^{t}\! \operatorname{Tr}\left[\sigma\!\left(s,T\right) \sqrt{X_{s}}\, dW_{s}\right]\,.
\end{equation}
\normalsize
In the sequel we use the following lemma.
\begin{lemma}\label{Example_Lemma}
 Let $\left(X_{t}\right)_{t\geq 0}$ be defined as in \eqref{Example_1}, $\sigma$ as in \emph{Assumption \ref{Assumption1}} and $f \in \mathcal{C}^{1}\!\left(\mathbb{R}\right)$. 
 Then, it holds for all $0 \leq t \leq T$:
 \begin{align*}\label{Example_Lemma_1}
              2\, \int_{0}^{t}\! f\!\left(s\right)\operatorname{Tr}\!\left[\sigma\!\left(s,T\right) \sqrt{X_{s}}\, dW_{s}\right] & = 
               - \int_{0}^{t}\! f\!\left(s\right)\left( \operatorname{Tr}\!\left[X_{s}\, \partial_{s}\sigma\!\left(s,T\right)\right] + a \operatorname{Tr}\!\left[\sigma\!\left(s,T\right)\right]\right) ds \nonumber\\
              & \phantom{==} + f\!\left(t\right) \operatorname{Tr}\!\left[\sigma\!\left(t,T\right) X_{t}\right] 
               -  f\!\left(0\right) \operatorname{Tr}\!\left[\sigma\!\left(0,T\right) X_{0}\right] \nonumber\\
              & \phantom{==} - \int_{0}^{t}\! f'\!\left(s\right) \operatorname{Tr}\!\left[\sigma\!\left(s,T\right) X_{s}\right] ds  \,.
 \end{align*}
\end{lemma}

\begin{proof}
 The proof can be found in \cite{Haertel_Thesis}.
\end{proof}

 We choose $\sigma\!\left(t,T\right)$ as a deterministic symmetric positive semidefinite $d \times d$ matrix 
in the following two different ways
\begin{itemize}
 \item[(i)]
\begin{equation}\label{Example_3}
 \sigma\!\left(t,T\right) \colonequals\begin{cases}
  \sigma e^{-\beta\left(T-t\right)} & \text{for } T \geq t\,,\\
  0 & \text{for } T < t\,,
\end{cases}
\end{equation}
 \item[(ii)] 
 \begin{equation}\label{Example_4}
 \sigma\!\left(t,T\right) \colonequals\begin{cases}
  \sigma \frac{1}{\sqrt{T-t}}  & \text{for } T > t\,,\\
  0 & \text{for } T \leq t\,.\nonumber
\end{cases}
 \end{equation}
\end{itemize}
 with $\sigma \in S_{d}^{+}$, $\sigma_{ij} \in \mathbb{R}$ for all $i,j \in \left\{1,\dots,d\right\}$, and $\beta > 0$.
 \par
 \noindent
 \vspace{0.01cm}
 \par\noindent
\textbf{Case (i):} Assumption \ref{Assumption1} is obviously fulfilled and Assumption \ref{Assumption2} follows for $0\leq s\leq t$ from
\begin{equation}\label{Example_5}
  \frac{1}{\sqrt{t}} \left\vert \Sigma\!\left(s,t\right)_{ij}\right\vert = \frac{ \left\vert \sigma_{ij} \right\vert}{\sqrt{t}}  \int_{s}^{t}\! e^{-\beta\left(u-s\right)}\, du 
  \leq \frac{2\sqrt{t-s}}{\beta \sqrt{t}}  \left\vert \sigma_{ij} \right\vert \leq \frac{2}{\beta} \left\vert \sigma_{ij} \right\vert \equalscolon w_{ij}\!\left(s\right)\,.\nonumber
 \end{equation}
Then, by Lemma \ref{Example_Lemma} we get for the forward rate
\begin{align}\label{Example_6}
f\!\left(t,T\right) & \overset{\eqref{Example_2}}{=} f\!\left(0,T\right) + 4 \int_{0}^{t}\! \operatorname{Tr}\left[\sigma\!\left(s,T\right) X_{s}\, \left(-\Sigma\!\left(s,T\right)\right)\right]ds + 2 \int_{0}^{t}\! \operatorname{Tr}\left[\sigma\!\left(s,T\right) \sqrt{X_{s}}\, dW_{s}\right]\nonumber\\
& \overset{\eqref{Sigma_Definition}}{=} f\!\left(0,T\right) + \frac{4}{\beta} \int_{0}^{t}\!\left(e^{-\beta\left(T-s\right)} - e^{-2\beta\left(T-s\right)}\right) \operatorname{Tr}\left[\sigma^{2} X_{s}\right] ds\nonumber\\
& \phantom{=====} + 2 \int_{0}^{t}\!e^{-\beta\left(T-s\right)} \operatorname{Tr}\left[\sigma \sqrt{X_{s}}\, dW_{s}\right]\nonumber\\
& \underset{\eqref{Trace_Representation}}{\overset{\eqref{Example_1}}{=}}f\!\left(0,T\right) + \frac{4}{\beta} \int_{0}^{t}\!\left(e^{-\beta\left(T-s\right)} - e^{-2\beta\left(T-s\right)}\right) \operatorname{Tr}\left[\sigma^{2} X_{s}\right] ds \nonumber\\
& \phantom{=====} - \beta \int_{0}^{t}\! e^{-\beta\left(T-s\right)} \operatorname{Tr}\!\left[\sigma X_{s}\right] ds  + e^{-\beta\left(T-t\right)} \operatorname{Tr}\!\left[\sigma X_{t}\right] \nonumber\\
& \phantom{=====} - e^{-\beta\, T}\operatorname{Tr}\!\left[\sigma X_{0}\right]
- \frac{a}{\beta} \operatorname{Tr}\!\left[\sigma\right] \left(e^{-\beta\left(T-t\right)} - e^{-\beta\, T}\right)\,.\nonumber
\end{align}
If we put 
\begin{equation}\label{Example_7}
 h_{0}\!\left(t,T\right) \colonequals f\!\left(0,T\right) - e^{-\beta T} \operatorname{Tr}\!\left[\sigma X_{0}\right] - \frac{a}{\beta} \operatorname{Tr}\!\left[\sigma\right] \left(e^{-\beta \left(T-t\right)} - e^{-\beta T}\right)\,,\nonumber
\end{equation}
\begin{equation}\label{Example_8}
 h\!\left(t,T\right) \colonequals 
  \begin{pmatrix}
  -\frac{4}{\beta} \sigma^{2} e^{-2 \beta T} \\ -\beta \sigma e^{- \beta T} + \frac{4}{\beta} \sigma^{2} e^{- \beta T} \\ e^{-\beta \left(T-t\right)} \sigma
 \end{pmatrix}\,, \nonumber
\end{equation}
and 
\begin{equation}\label{Example_9}
 Z_{t}\colonequals 
  \begin{pmatrix}
 \int_{0}^{t}\! e^{2 \beta s} X_{s} \, ds \\ \int_{0}^{t}\! e^{\beta s} X_{s} \, ds \\ X_{t} 
 \end{pmatrix}\,,
\end{equation}
we get
\begin{equation}\label{Example_10}
 f\!\left(t,T\right) = h_{0}\!\left(t,T\right) + h\!\left(t,T\right) \cdot Z_{t} = h_{0}\!\left(t,T\right) + \operatorname{Tr}\!\left[Z_{t}^{\top} h\!\left(t,T\right)\right]\,.\nonumber
\end{equation}
Note that in this setting we obtain an affine $d$-dimensional realisation for the forward rate as in Definition 3 of \cite{Chiarella_Kwon} with $h_{0}\!\left(t,T\right)$ and $h\!\left(t,T\right)$ being $\mathcal{F}_{t}$-measurable 
processes and $Z$ an affine process.
\par
Furthermore, the short rate process $\left(r_{t}\right)_{t\geq 0}$ has the form 
\begin{equation}\label{Example_14}
 r_{t} \overset{\eqref{Short_Rate_Process_Equation}}{=} r_{0} - \operatorname{Tr}\!\left[\sigma\, X_{0}\right] + \int_{0}^{t}\! \phi\!\left(u\right)\, du + \operatorname{Tr}\!\left[\sigma\, X_{t}\right]\,.\nonumber
\end{equation}
with $\phi$ as in \eqref{Phi_Definition}, since $\int_{0}^{t} \operatorname{Tr}\!\left[\sigma\!\left(u,u\right) dX_{u}\right] = \operatorname{Tr}\!\left[\sigma\left(X_{t} - X_{0}\right)\right]$ 
for all $t \geq 0$ by \eqref{Example_3}. 
By Theorem \ref{HJM_Drift_Condition_Theorem} we get that
\begin{align*}
r_{t} & \overset{\eqref{HJM_Drift_Condition_Equation}}{=} f\!\left(0,t\right) + \frac{1}{\beta} \left(1 - \beta - e^{-\beta t}\right) \operatorname{Tr}\!\left[\sigma X_{0}\right]
- \frac{a}{\beta} \left(1 - e^{-\beta t}\right)\operatorname{Tr}\!\left[\sigma\right] + \operatorname{Tr}\!\left[\sigma X_{t}\right] \nonumber\\
& \phantom{====} - \beta \int_{0}^{t}\! e^{-\beta \left(t-s\right)} \operatorname{Tr}\!\left[\sigma X_{s}\right]\, ds - \frac{4}{\beta} \int_{0}^{t}\! \left(e^{-2 \beta \left(t - s\right)} - e^{-\beta \left(t - s\right)}\right)\operatorname{Tr}\!\left[\sigma^{2} X_{s}\right]\, ds\,,
\end{align*}
hence the short-rate process $r$ is a Markov process with respect to $Z$, defined in \eqref{Example_9}.
\par
\noindent
The yield has the following form 
\small
\begin{align}\label{Example_11}
Y\!\left(t,T\right) & \overset{\eqref{Cont_Comp_Spot_Rate_Equation}}{=} 
Y\!\left(0;t,T\right) + \frac{2}{\left(T-t\right)\beta^{2}} \int_{0}^{t} \! \left(\left(e^{-\beta\left(T-s\right)} - 1\right)^{2} - \left(e^{-\beta\left(t-s\right)} - 1\right)^{2}\right) \operatorname{Tr}\!\left[\sigma^{2} X_{s}\right]ds\nonumber\\
& \phantom{====} -\frac{1}{T\!-\!t}\left( \int_{0}^{t}\!  \left(e^{-\beta\left(T-s\right)} - e^{-\beta\left(t-s\right)}\right)\operatorname{Tr}\!\left[\sigma X_{s}\right] ds + \frac{1}{\beta} \left(1 - e^{-\beta\left(T-t\right)}\right) \operatorname{Tr}\!\left[\sigma X_{t}\right] \right)\nonumber\\
& \phantom{====} - \frac{1}{\left(T\!-\!t\right)\beta}\left( \left(e^{-\beta T} -\! e^{-\beta t}\right) \operatorname{Tr}\!\left[\sigma X_{0}\right] + \frac{a \operatorname{Tr}\!\left[\sigma\right]}{\beta} \left( e^{-\beta\left(T-t\right)} -\!  e^{\beta T} +\! e^{-\beta t} -\! 1\right)\right)\,. \nonumber
\end{align}
\normalsize
Then, the long-term drift for $t \geq 0$ is
\begin{equation}\label{Example_12}
 \mu_{\infty}\!\left(t\right) \overset{\eqref{LongTermDriftEqu}}{=} \lim\limits_{T \rightarrow \infty} \frac{\Gamma\!\left(t,T\right)}{T} 
 \overset{\eqref{Definition_Gamma}}{=} \lim\limits_{T \rightarrow \infty} \frac{1}{\beta^{2} T} \left(e^{-\beta \left(T-t\right)} - 1\right)^{2} \sigma\, X_{t}\, \sigma = 0\,,
\end{equation}
and we get by Theorem \ref{LongTermYield} that 
\begin{align}\label{Example_13}
\ell_{t} & \overset{\eqref{LongTermYieldEqu}}{=} \ell_{0} + 2 \int_{0}^{t} \! \operatorname{Tr}\!\left[\mu_{\infty}\!\left(s\right)\right]ds 
\overset{\eqref{Example_12}}{=} \ell_{0} \,,\nonumber
\end{align}
i.e.\,\,$\left(\ell_{t}\right)_{t \geq 0}$ is constant.
\par\noindent
 \vspace{0.01cm}
 \par\noindent
\textbf{Case (ii):} 
Assumption \ref{Assumption1} is again fulfilled obviously and Assumption \ref{Assumption2} follows for $0\leq s\leq t$ since
\begin{equation}\label{Example_2_1}
  \frac{1}{\sqrt{t}} \left\vert \Sigma\!\left(s,t\right)_{ij}\right\vert = 2 \left\vert \sigma_{ij} \right\vert \frac{\sqrt{t-s}}{\sqrt{t}} 
  = 2 \left\vert \sigma_{ij} \right\vert \sqrt{1 - \frac{s}{t}} \leq 2 \left\vert \sigma_{ij} \right\vert \equalscolon w_{ij}\!\left(s\right)\,.\nonumber
 \end{equation}
The forward rate for $0 \leq t \leq T$ is
\begin{align*}
f\!\left(t,T\right) & \overset{\eqref{Example_2}}{=} f\!\left(0,T\right) + 4 \int_{0}^{t}\! \operatorname{Tr}\left[\sigma\!\left(s,T\right) X_{s}\, \left(-\Sigma\!\left(s,T\right)\right)\right]ds + 2 \int_{0}^{t}\! \operatorname{Tr}\left[\sigma\!\left(s,T\right) \sqrt{X_{s}}\, dW_{s}\right]\nonumber\\
& \overset{\eqref{Sigma_Definition}}{=}f\!\left(0,T\right) + 8 \int_{0}^{t}\! \operatorname{Tr}\left[\sigma^{2}\, X_{s}\right] ds 
- \frac{1}{2} \int_{0}^{t}\! \left(T\!-\!s\right)^{-\frac{3}{2}} \operatorname{Tr}\!\left[\sigma X_{s}\right] ds\nonumber\\
& \phantom{=====} +  \frac{1}{\sqrt{T\!-\!t}} \operatorname{Tr}\!\left[\sigma X_{t}\right] - \frac{1}{\sqrt{T}} \operatorname{Tr}\!\left[\sigma X_{0}\right] 
- 2 a \operatorname{Tr}\!\left[\sigma\right] \left(\sqrt{T} - \sqrt{T\!-\!t}\right) \nonumber
\end{align*}
and the yield for $\left[t,T\right]$ has the following forms
\small
\begin{align*}
Y\!\left(t,T\right) & \overset{\eqref{Cont_Comp_Spot_Rate_Equation}}{=} 
Y\!\left(0;t,T\right) + 8 \int_{0}^{t} \!  \operatorname{Tr}\left[\sigma^{2}\, X_{s}\right] ds + \frac{1}{T\!-\!t} \int_{0}^{t}\! \left(\frac{1}{\sqrt{T\!-\!s}} - \frac{1}{\sqrt{t\!-\!s}}\right)\operatorname{Tr}\!\left[\sigma X_{s}\right] ds \nonumber\\
& \phantom{====} + \frac{2\, \operatorname{Tr}\!\left[\sigma X_{t}\right]}{\sqrt{T\!-\!t}}  - \frac{2 \left(\sqrt{T}\! -\! \sqrt{t}\right) \operatorname{Tr}\!\left[\sigma X_{0}\right]}{T\!-\!t} + \frac{4}{3} a \operatorname{Tr}\!\left[\sigma\right] \left(\sqrt{T\!-\!t} - \frac{T\sqrt{T} - t\sqrt{t}}{T\!-\!t}\right)\,.\nonumber
\end{align*}
\normalsize
Then, the long-term drift for $t \geq 0$ is
\begin{equation}\label{Example_2_4}
 \mu_{\infty}\!\left(t\right) \overset{\eqref{LongTermDriftEqu}}{=} \lim\limits_{T \rightarrow \infty} \frac{\Gamma\!\left(t,T\right)}{T} 
 \overset{\eqref{Definition_Gamma}}{=} \lim\limits_{T \rightarrow \infty} \frac{4 \left(T-t\right) \sigma X_{t} \sigma}{T} = 4 \sigma\, X_{t}\, \sigma\,,
\end{equation}
and we get with Theorem \ref{LongTermYield} that 
\begin{align}\label{Example_2_5}
\ell_{t} & \overset{\eqref{LongTermYieldEqu}}{=} \ell_{0} + 2 \int_{0}^{t} \! \operatorname{Tr}\!\left[\mu_{\infty}\!\left(s\right)\right]ds 
\overset{\eqref{Example_2_4}}{=} \ell_{0} + 8 \int_{0}^{t} \! \operatorname{Tr}\!\left[\sigma\, X_{s}\, \sigma\right]ds\,.\nonumber
\end{align}



\appendix

\section{}

\addcontentsline{toc}{section}{A} 

Here we provide the proofs of Proposition \ref{Bond_Price_Process_Proposition} and Theorem \ref{HJM_Drift_Condition_Theorem}. 
The results follow by applying the Fubini theorem for integrable functions (cf.\,\,Theorem 14.16 in Chapter 14 of \cite{Book_Klenke}) 
and the stochastic Fubini theorem (cf.\,\,Theorem 65 in Chapter IV of \cite{Book_Protter}).
For further details on the following computations, we refer to \cite{Haertel_Thesis}.

\begin{pf_appendix_1}
 \emph{Let us introduce for every maturity $T > 0$ the quantity}
\begin{equation}\label{Forward_Rate_Integral_Definition}
Z\!\left(t,T\right) \colonequals -\int_{t}^{T} f\!\left(t,u\right) \, du\,,
\end{equation}
\emph{for all $0 \leq t \leq T$.}
\emph{From the dynamics of the forward rate \eqref{Forward_Rates_Process_Definition} we deduce that for all $T > 0$}
\begin{equation}\label{Forward_Rate_Integral_Representation}
Z\!\left(t,T\right) \overset{\eqref{Forward_Rate_Integral_Definition}}{=} -\int_{t}^{T} \! f\!\left(0,u\right)\, du - \int_{t}^{T} \! \int_{0}^{t} \! \alpha\!\left(s,u\right)\,  ds \, du - \int_{t}^{T} \! \int_{0}^{t} \! \operatorname{Tr}\!\left[\sigma\!\left(s,u\right) dX_s \right] du\,,
\end{equation}
\emph{for all $0 \leq t \leq T$.
By combining \eqref{Affine_Process_Representation_Equation_1}, \eqref{Sigma_Definition}, \eqref{Forward_Rate_Integral_Representation}, and \eqref{Short_Rate_Representation}, 
we derive the following identity}
\begin{align*}
Z\!\left(t,T\right) & \overset{(\ref{Forward_Rate_Integral_Representation})}{=} Z\!\left(0,T\right) \!+\! \int_{0}^{t} \! f\!\left(0,u\right) \, du \!-\! \int_{t}^{T} \! \int_{0}^{t} \! \alpha\!\left(s,u\right) \, ds \, du \!-\! \int_{t}^{T} \! \int_{0}^{t} \! \operatorname{Tr}\!\left[ \sigma\!\left(s,u\right) dX_{s} \right] du\nonumber\\
& \overset{\eqref{Short_Rate_Representation}}{\underset{\phantom{\eqref{Affine_Process_Representation_Equation_1}}}{=}} Z\!\left(0,T\right) + \int_{0}^{t} \! r_{s} \, ds -\int_{0}^{t} \! \int_{s}^{T}\! \alpha\!\left(s,u\right) \, du \, ds - \int_{0}^{t} \! \int_{s}^{T}\! \operatorname{Tr}\!\left[\sigma\!\left(s,u\right) \, du \, dX_{s} \right]\nonumber\\
& \overset{\eqref{Affine_Process_Representation_Equation_1}}{\underset{\eqref{Sigma_Definition}}{=}} Z\!\left(0,T\right) + \int_{0}^{t} \! r_{s} \, ds -\int_{0}^{t} \! \int_{s}^{T}\! \alpha\!\left(s,u\right) \, du \, ds \nonumber\\
& \phantom{====} +\! \int_{0}^{t} \!\operatorname{Tr}\!\left[\Sigma\!\left(s,T\right)\left(\sqrt{X_{s}}\,dW_{s}\,Q + Q^{\top} dW_{s}^{\top} \sqrt{X_{s}}\right)\right] \nonumber\\
& \phantom{====} +\! \int_{0}^{t} \operatorname{Tr}\!\left[\Sigma\!\left(s,T\right)\left(b +\! B\!\left(X_{s}\right)\right)\right] ds +\! \int_{0}^{t}\!\int_{S_{d}^{+}\! \setminus \left\{ 0\right\}}\!\operatorname{Tr}\!\left[\Sigma\!\left(s,T\right)\xi\right] \mu^{X\!}\!\left(ds,d\xi\right)\,.
\end{align*}
\emph{Note that in general for $A, B \in \mathcal{M}_{d}$ with $A$ symmetric, i.e. $A \in S_{d}$, it holds}
\begin{align}\label{Trace_Representation}
 \operatorname{Tr}\!\left[ A \left(B + B^{\top}\right)\right] & = \operatorname{Tr}\!\left[A B \right] + \operatorname{Tr}\!\left[A B^{\top} \right] = \operatorname{Tr}\!\left[A B \right] + \operatorname{Tr}\!\left[\left(B A\right)^{\top} \right]\nonumber\\
&  = \operatorname{Tr}\!\left[A B \right] + \operatorname{Tr}\left[ B A \right] = 2 \operatorname{Tr}\!\left[A B \right]\,.
\end{align}
\emph{Therefore, we get due to $\sigma\!\left(s,t\right) \in S_{d}$ for all $s, t \geq 0$ that}
\begin{align}\label{Forward_Rate_Integral_Dynamics_2}
Z\!\left(t,T\right) & \overset{(\ref{Trace_Representation})}{=} Z\!\left(0,T\right) \!+\! \int_{0}^{t} \! r_{s} \, ds \!-\!\int_{0}^{t} \! \int_{s}^{T}\! \alpha\!\left(s,u\right) \, du \, ds \!+\! 2 \int_{0}^{t} \! \operatorname{Tr}\!\left[ \Sigma\!\left(s,T\right) \sqrt{X_{s}}\, dW_{s}\, Q \right] \nonumber\\
 & \phantom{====} +\! \int_{0}^{t} \operatorname{Tr}\!\left[\Sigma\!\left(s,T\right)\left(b +\! B\!\left(X_{s}\right)\right)\right] ds +\! \int_{0}^{t}\!\int_{S_{d}^{+}\! \setminus \left\{ 0\right\}}\!\operatorname{Tr}\!\left[\Sigma\!\left(s,T\right)\xi\right] \mu^{X\!}\!\left(ds,d\xi\right)\,.
\end{align}
\emph{Note that for all $0 \leq t \leq T$}
\begin{equation}\label{Delta_Z}
 \Delta Z\!\left(t,T\right) = \operatorname{Tr}\left[\Sigma\!\left(t,T\right)\Delta X_{t}\right]\,.
\end{equation}
\emph{With the help of \eqref{Forward_Rate_Integral_Dynamics_2} and the fact that}
\begin{equation}\label{Brownian_Motion_Dynamics}
 \left\langle W_{lm}, W_{ru} \right\rangle_{s} = \left\{ \begin{array}{lc} s & \text{\emph{ if} } l = r \text{\emph{ and} } m = u,\\
                                     0 & \text{\emph{ else,}}
                                    \end{array}\right.
\end{equation}
\emph{we can calculate the quadratic variation of $Z$ for all $T > 0$ as follows}
\small
\begin{align}\label{Quadratic_Variation_Forward_Rate_Integral}
  \left\langle Z\!\left(\,\cdot\,,T\right) \right\rangle_{t}^{c} & \overset{\phantom{\eqref{Brownian_Motion_Dynamics}}}{=}  \left\langle \operatorname{Tr}\!\left[\int_{0}^{\cdot}\! \Sigma\!\left(s,T\right) \sqrt{X_{s}}\, dW_{s}\, Q\right] \right\rangle_{t} 
 \overset{\eqref{Brownian_Motion_Dynamics}}{=} 4 \int_{0}^{t}\! \operatorname{Tr}\!\left[ Q\, \Sigma\!\left(s,T\right) X_{s}\, \Sigma\!\left(s,T\right) Q^{\top}\right]\, ds\,.
\end{align}
\normalsize
\emph{Further, we see that due to equation (2.27) of \cite{article_Cuchiero} for all $u \in S_{d}^{+}$ and a process $Y$ on $S_{d}^{+}$ it holds}
\begin{equation}\label{Remark_Bond_Price_Equation_2}
 \operatorname{Tr}\!\left[B^{\top}\!\left(u\right) Y \right] = \operatorname{Tr}\!\left[B\!\left(Y\right) u \right]\,,
\end{equation}
\emph{where $B$ is defined according to \eqref{Linear_Drift_Equation_2} and therefore for all $0\leq t\leq T$ and $\alpha = Q^{\top}Q$}
\begin{align}\label{Help_F_R}
& \int_{0}^{t}\!\operatorname{Tr}\!\left[\Sigma\!\left(s,T\right)\left(b + B\!\left(X_{s}\right)\right)\right] + 2\, \operatorname{Tr}\!\left[Q\, \Sigma\!\left(s,T\right) X_{s}\, \Sigma\!\left(s,T\right) Q^{\top} \right]ds \nonumber\\
& \phantom{========} +  \int_{0}^{t}\!\int_{S_{d}^{+}\! \setminus \left\{ 0\right\}}\!\left(e^{\operatorname{Tr}\left[\Sigma\left(s,T\right)\, \xi\right]} - 1\right)\, \nu\!\left(ds,d\xi\right)\nonumber\\
& \phantom{====}\underset{\eqref{Remark_Bond_Price_Equation_2}}{\overset{\eqref{Affine_Process_Representation_Equation_2}}{=}}-\int_{0}^{t}\!\operatorname{Tr}\!\left[-\Sigma\!\left(s,T\right) b - 2\, \Sigma\!\left(s,T\right) \alpha \Sigma\!\left(s,T\right) X_{s} + B^{\top}\!\left(-\Sigma\!\left(s,T\right)\right) X_{s}\right]ds \nonumber\\
& \phantom{========} + \int_{0}^{t}\!\int_{S_{d}^{+}\! \setminus \left\{ 0\right\}}\!\left(e^{-\operatorname{Tr}\left[-\Sigma\left(s,T\right)\, \xi\right]} - 1\right)\, m\!\left(d\xi\right)ds\nonumber\\
& \phantom{========} + \int_{0}^{t}\!\operatorname{Tr}\!\left[ X_{s} \int_{S_{d}^{+}\! \setminus \left\{ 0\right\}}\!\left(e^{-\operatorname{Tr}\left[-\Sigma\left(s,T\right)\, \xi\right]} - 1\right)\, \mu\!\left(d\xi\right)\right]ds \nonumber\\
& \phantom{====}\underset{\eqref{Riccati_Equations_Theorem_Equation_3}}{\overset{\eqref{Riccati_Equations_Theorem_Equation_4}}{=}} \int_{0}^{t}\!\left(-F\!\left(-\Sigma\!\left(s,T\right)\right)-\operatorname{Tr}\!\left[R\!\left(-\Sigma\!\left(s,T\right)\right)X_{s}\right] \right)ds\,.
\end{align}
\emph{Now, we apply It\^{o}'s formula on $P\!\left(t,T\right) \colonequals \exp\!\left(Z\!\left(t,T\right)\right)$ for every maturity $T > 0$ (cf.\,\,Definition 1.4.2 of \cite{brigo06}) and obtain the following representation, where 
we use Proposition 1.28 of Chapter II in \cite{Book_JacodShiryaev} to combine the measures  $\mu^{X\!}\!\left(ds,d\xi\right)$ and $\nu\!\left(ds,d\xi\right)$, 
since the affine process $X$ has only jumps of finite variation (cf.\,\,\eqref{Finite_Variation_Jumps_Remark_Equation}).}  
\small
\begin{align*}
P\!\left(t,T\right) & \underset{\phantom{\eqref{Quadratic_Variation_Forward_Rate_Integral}}}{\overset{\phantom{(\ref{A_Definition})}}{=}} P\!\left(0,T\right) + \int_{0}^{t} \! P\!\left(s-,T\right)\, dZ\!\left(s,T\right) + \frac{1}{2}\int_{0}^{t} P\!\left(s,T\right) \, d\left\langle Z\!\left(\,\cdot\,,T\right)\right\rangle^{c}_{s}\nonumber\\
 & \phantom{\underset{\eqref{A_Definition}}{\overset{(\ref{A_Definition})}{=}}==} + \sum_{0 < s \leq t}^{\Delta Z\left(s,T\right)\neq 0} \left[e^{Z\left(s,T\right)} - e^{Z\left(s-,T\right)} - \Delta Z\!\left(s,T\right)e^{Z\left(s-,T\right)}\right]\nonumber\\
 & \overset{\eqref{Delta_Z}}{\underset{\eqref{Quadratic_Variation_Forward_Rate_Integral}}{=}} P\!\left(0,T\right) + \int_{0}^{t} \! P\!\left(s-,T\right)\, dZ\!\left(s,T\right) \nonumber\\
 & \phantom{\underset{(\ref{A_Definition})}{\overset{(\ref{A_Definition})}{=}}==} + 2 \int_{0}^{t} \! P\!\left(s,T\right) \operatorname{Tr}\!\left[ Q\, \Sigma\!\left(s,T\right) X_{s}\, \Sigma\!\left(s,T\right) Q^{\top} \right]\, ds\nonumber\\
 & \phantom{\underset{(\ref{A_Definition})}{\overset{(\ref{A_Definition})}{=}}==} + \sum_{0 \leq s \leq t}^{\Delta X_{s}\neq 0} \left[e^{Z\left(s,T\right)} - e^{Z\left(s-,T\right)} - \operatorname{Tr}\!\left[\Sigma\!\left(s,T\right) \Delta X_{s}\right] e^{Z\left(s-,T\right)}\right]\nonumber\\
 & \overset{\eqref{Forward_Rate_Integral_Dynamics_2}}{\underset{\eqref{Jump_Definition}}{=}}  P\!\left(0,T\right) + 2 \int_{0}^{t} \! P\!\left(s,T\right) \operatorname{Tr}\!\left[ \Sigma\!\left(s,T\right) \sqrt{X_{s}}\, dW_{s}\, Q \right]\nonumber\\
 & \phantom{\underset{(\ref{A_Definition})}{\overset{(\ref{A_Definition})}{=}}==} + \int_{0}^{t}\! P(s,T)\left(r_{s} - \int_{s}^{T}\!\alpha\!\left(s,u\right) du\right)\, ds \nonumber\\
 & \phantom{\underset{(\ref{A_Definition})}{\overset{(\ref{A_Definition})}{=}}==} + \int_{0}^{t}\! P(s,T) \operatorname{Tr}\!\left[\Sigma\!\left(s,T\right)\left(b + B\!\left(X_{s}\right)\right)\right]\, ds\nonumber\\
 & \phantom{\underset{(\ref{A_Definition})}{\overset{(\ref{A_Definition})}{=}}==} + \int_{0}^{t}\! P(s-,T) \int_{S_{d}^{+}\! \setminus \left\{ 0\right\}}\!\operatorname{Tr}\!\left[\Sigma\!\left(s,T\right) \xi\right] \, \mu^{X\!}\!\left(ds,d\xi\right)\nonumber\\
 & \phantom{\underset{(\ref{A_Definition})}{\overset{(\ref{A_Definition})}{=}}==} + 2 \int_{0}^{t} \! P\!\left(s,T\right) \operatorname{Tr}\!\left[ Q\, \Sigma\!\left(s,T\right) X_{s}\, \Sigma\!\left(s,T\right) Q^{\top} \right]\, ds\nonumber\\ 
 & \phantom{\underset{(\ref{A_Definition})}{\overset{(\ref{A_Definition})}{=}}==} + \sum_{0 \leq s \leq t}^{\Delta X_{s}\neq 0} \left[e^{\Delta Z\left(s,T\right)}P\!\left(s-,T\right)\! -\! P\!\left(s-,T\right)\! -\! \operatorname{Tr}\!\left[\Sigma\!\left(s,T\right) \Delta X_{s}\right] P\!\left(s-,T\right)\right] \nonumber\\
 & \underset{\eqref{Help_F_R}}{\overset{\eqref{Delta_Z}}{=}} P\!\left(0,T\right) + \int_{0}^{t}\! P\!\left(s-,T\right)\left(r_{s} + A\!\left(s,T\right)\right) ds 
 + 2 \int_{0}^{t}\! P\!\left(s,T\right) \operatorname{Tr}\!\left[\Sigma\!\left(s,T\right)\sqrt{X_{s}} \, dW_{s} Q\right] \nonumber\\
 & \phantom{\underset{(\ref{A_Definition})}{\overset{(\ref{A_Definition})}{=}}==} + \int_{0}^{t}\! P\!\left(s-,T\right) \int_{S_{d}^{+}\! \setminus \left\{ 0\right\}}\! \left(e^{\operatorname{Tr}\left[ \Sigma\left(s,T\right)\,\xi\right]} - 1\right) \left(\mu^{X} - \nu\right)\left(ds,d\xi\right)\,.
\end{align*}
\normalsize
\emph{Assumption \ref{Assumption1} guarantees that all integrals above are finite.}
\hfill
$\Box$
\end{pf_appendix_1}

 \begin{pf_appendix_2}
\emph{By using (\ref{Discounted_Bond_Price_Process_Corollary}) we see that the discounted bond price process under $\mathbb{Q}$ is}
\begin{align}\label{Discounted_Bond_Price_Process_Representation_Q}
 \frac{P\!\left(t,T\right)}{\beta_{t}} 
& \overset{\emph{\eqref{Q_Compensator_Equation_2}}}{=} P\!\left(0,T\right) + \int_{0}^{t}\! \frac{P\!\left(s,T\right)}{\beta_{s}}\left( A\!\left(s,T\right) + 2 \operatorname{Tr}\!\left[\Sigma\!\left(s,T\right)\! \sqrt{X_{s}}\, \gamma_{s}\, Q\right]\right) ds \nonumber\\
 & \phantom{=====} + 2 \int_{0}^{t}\! \frac{P\!\left(s,T\right)}{\beta_{s}} \operatorname{Tr}\!\left[\Sigma\!\left(s,T\right) \sqrt{X_{s}}\, dW^{\ast}_{s} \, Q\right]\nonumber\\
 & \phantom{=====} + \int_{0}^{t}\!\int_{S_{d}^{+}\! \setminus \left\{ 0\right\}}\! \frac{P\!\left(s-,T\right)}{\beta_{s}}\left(e^{\operatorname{Tr}\left[\Sigma\left(s,T\right)\,\xi\right]}\! -\! 1\right)\left(\mu^{X\!} - \nu^{\ast\!}\right)\!\left(ds,d\xi\right)\nonumber\\
 & \phantom{=====} + \int_{0}^{t}\!\int_{S_{d}^{+}\! \setminus \left\{ 0\right\}}\! \frac{P\!\left(s-,T\right)}{\beta_{s}}\left(e^{\operatorname{Tr}\left[\Sigma\left(s,T\right)\,\xi\right]}\! -\! 1\right)\left(K\!\left(s,\xi\right)\!-\!1\right)\nu\!\left(ds,d\xi\right)
\end{align}
\emph{for all $0\leq t\leq T$.} 
\emph{Since $\frac{P\left(t,T\right)}{\beta_{t}}$, $t \leq T$, has to be a local martingale under $\mathbb{Q}$, the drift in \eqref{Discounted_Bond_Price_Process_Representation_Q}
 must disappear, i.e.\,\,for all $0\leq t\leq T$}\small
\begin{align}
 0 & \overset{\emph{\eqref{Affine_Process_Representation_Equation_2}}}{=} \int_{0}^{t}\! \frac{P\!\left(s,T\right)}{\beta_{s}} A\!\left(s,T\right)\, ds + 2 \int_{0}^{t}\! \frac{P\!\left(s,T\right)}{\beta_{s}} \operatorname{Tr}\!\left[\Sigma\!\left(s,T\right)\! \sqrt{X_{s}}\, \gamma_{s}\, Q\right] ds \nonumber\\
& \phantom{===} + \int_{0}^{t}\!\int_{S_{d}^{+}\! \setminus \left\{ 0\right\}}\! \frac{P\!\left(s-,T\right)}{\beta_{s}}\left(e^{\operatorname{Tr}\left[\Sigma\left(s,T\right)\,\xi\right]}\! -\! 1\right)\left(K\!\left(s,\xi\right)\!-\!1\right)\left(m\!\left(d\xi\right)\! +\! \operatorname{Tr}\!\left[X_{s} \mu\!\left(d\xi\right)\right]\right) ds. \nonumber
\end{align}
 \normalsize
\emph{It follows for all $0\leq t\leq T$ that}
\begin{align*}
 A\!\left(t,T\right) & = - 2 \operatorname{Tr}\!\left[\Sigma\!\left(t,T\right)\! \sqrt{X_{t}}\, \gamma_{t}\, Q\right] \nonumber\\
& \phantom{=====} - \int_{S_{d}^{+}\! \setminus \left\{ 0\right\}\!}\!\left(e^{\operatorname{Tr}\left[\Sigma\left(t,T\right)\,\xi\right]}\! -\! 1\right)\left(K\!\left(t,\xi\right)\!-\!1\right) \left(m\!\left(d\xi\right)\! +\! \operatorname{Tr}\!\left[X_{t} \mu\!\left(d\xi\right)\right]\right)
\end{align*}
$dt \otimes d\mathbb{P}$-a.s.
\par
\noindent
\emph{Consequently we get for all $0\leq t\leq T$ by Satz 6.28 in \cite{Book_Klenke}}
\begin{align*}\label{Alpha_HJM_Drift}
 \alpha\!\left(t,T\right) & \underset{\phantom{\eqref{Discounted_Bond_Price_Process_Representation_Q}}}{\overset{\emph{\eqref{A_Definition}}}{=}} - \partial_{T} A\!\left(t,T\right) - \partial_{T} F\!\left(-\Sigma\!\left(t,T\right)\right) - \partial_{T} \operatorname{Tr}\!\left[R\!\left(-\Sigma\!\left(t,T\right)\right)X_{t}\right] \nonumber\\
& \underset{\phantom{\eqref{Discounted_Bond_Price_Process_Representation_Q}}}{\overset{\emph{\eqref{Trace_Representation}}}{=}} -\! \operatorname{Tr}\!\left[\sigma\!\left(t,T\right)\left(b + B\!\left(X_{t}\right) + 2 \sqrt{X_{t}}\, \gamma_{t}\, Q\right)\right]\! -\! 4 \operatorname{Tr}\!\left[Q\, \sigma\!\left(t,T\right) X_{t}\, \Sigma\!\left(t,T\right) Q^{\top}\right] \nonumber\\
& \phantom{====} - \int_{S_{d}^{+}\! \setminus \left\{ 0\right\}\!}\! \operatorname{Tr}\left[\sigma\!\left(t,T\right) \xi\right] e^{\operatorname{Tr}\left[\Sigma\left(t,T\right)\,\xi\right]} K\!\left(t,\xi\right) \left(m\!\left(d\xi\right)\! +\! \operatorname{Tr}\!\left[X_{t} \mu\!\left(d\xi\right)\right]\right)\nonumber
\end{align*}
\emph{$dt \otimes d\mathbb{P}$-a.s.}
\par
\noindent
\emph{
Hence, $\frac{P\left(t,T\right)}{\beta_{t}}$, $t \leq T$, is a $\mathbb{Q}$-local martingale if and only if equation (\ref{HJM_Drift_Condition_Equation}) is fulfilled $dt \otimes d\mathbb{P}$-a.s.
Equation \eqref{HJM_Drift_Condition_Equation} represents the HJM condition on the drift in the affine setting on $S_{d}^{+}$.
Then, the forward rate under $\mathbb{Q}$ follows a process of the form}
\small
\begin{align*}
 f\!\left(t,T\right) & \underset{\eqref{Affine_Process_Representation_Equation_1}}{\overset{\eqref{Forward_Rates_Process_Definition}}{=}} f\!\left(0,T\right) + \int_{0}^{t}\! \alpha\!\left(s,T\right) ds + \int_{0}^{t}\! \operatorname{Tr}\!\left[\sigma\!\left(s,T\right)\left(b +  B\!\left(X_{s}\right) + 2 \sqrt{X_{s}}\, \gamma_{s}\, Q\right)\right] ds \nonumber\\
& \phantom{=====} + \int_{0}^{t}\!\int_{S_{d}^{+}\! \setminus \left\{ 0\right\}}\! \operatorname{Tr}\!\left[\sigma\!\left(s,T\right) \xi\right]\, \mu^{X\!}\!\left(ds,d\xi\right) + 2 \int_{0}^{t}\! \operatorname{Tr}\left[\sigma\!\left(s,T\right) \sqrt{X_{s}}\, dW^{\ast}_{s}\, Q\right] \nonumber\\
& \overset{\emph{\eqref{HJM_Drift_Condition_Equation}}}{=} f\!\left(0,T\right) - 4 \int_{0}^{t}\! \operatorname{Tr}\left[Q\, \sigma\!\left(s,T\right)\, X_{s}\, \Sigma\!\left(s,T\right)\, Q^{\top}\right]\,ds\nonumber\\
& \phantom{=====} - \int_{0}^{t}\!\int_{S_{d}^{+}\! \setminus \left\{ 0\right\}}\!\operatorname{Tr}\!\left[\sigma\!\left(s,T\right) \xi\right] e^{\operatorname{Tr}\left[\Sigma\left(s,T\right)\,\xi\right]}K\!\left(s,\xi\right) \nu\!\left(ds,d\xi\right)\nonumber\\
& \phantom{=====} + \int_{0}^{t}\!\int_{S_{d}^{+}\! \setminus \left\{ 0\right\}}\! \operatorname{Tr}\!\left[\sigma\!\left(s,T\right) \xi\right]\, \mu^{X\!}\!\left(ds,d\xi\right) + 2 \int_{0}^{t}\! \operatorname{Tr}\left[\sigma\!\left(s,T\right) \sqrt{X_{s}}\, dW^{\ast}_{s}\, Q\right]\nonumber\\
& \underset{\phantom{\eqref{HJM_Drift_Condition_Equation}}}{\overset{\emph{\eqref{Sigma_Definition}}}{=}} f\!\left(0,T\right) + 4 \int_{0}^{t}\! \operatorname{Tr}\left[ Q\, \sigma\!\left(s,T\right) X_{s}\, \int_{s}^{T}\!\sigma\!\left(s,u\right) du\  Q^{\top}\right]\,ds\nonumber\\
& \phantom{=====} + \int_{0}^{t}\!\int_{S_{d}^{+}\! \setminus \left\{ 0\right\}}\!\operatorname{Tr}\!\left[\sigma\!\left(s,T\right) \xi\right] \left(\mu^{X\!}-\nu^{\ast\!}\right)\!\left(ds,d\xi\right)\nonumber\\
& \phantom{=====} - \int_{0}^{t}\!\int_{S_{d}^{+}\! \setminus \left\{ 0\right\}}\!\operatorname{Tr}\!\left[\sigma\!\left(s,T\right) \xi\right]\left(e^{\operatorname{Tr}\left[\Sigma\left(s,T\right)\,\xi\right]} - 1\right)\nu^{\ast\!}\!\left(ds,d\xi\right)\nonumber\\
& \phantom{=====} + 2 \int_{0}^{t}\! \operatorname{Tr}\left[\sigma\!\left(s,T\right) \sqrt{X_{s}}\, dW^{\ast}_{s}\, Q\right]\nonumber\\
& \underset{\phantom{\emph{\eqref{HJM_Drift_Condition_Equation}}}}{\overset{\emph{\eqref{Affine_Process_Representation_Equation_2}}}{=}} f\!\left(0,T\right) + \int_{0}^{t}\! \left\{4 \operatorname{Tr}\left[ Q\, \sigma\!\left(s,T\right) X_{s}\, \int_{s}^{T}\!\sigma\!\left(s,u\right) du\  Q^{\top}\right] \right.\nonumber\\
& \phantom{=====} \left. - \int_{S_{d}^{+}\! \setminus \left\{ 0\right\}\!}\! K\!\left(s,\xi\right) \operatorname{Tr}\!\left[\sigma\!\left(s,T\right) \xi\right] \left(e^{\operatorname{Tr}\left[\Sigma\left(s,T\right)\,\xi\right]}\! -\! 1\right) \left(m\!\left(d\xi\right)\! +\! \operatorname{Tr}\!\left[X_{s} \mu\!\left(d\xi\right)\right] \right)\right\} ds \displaybreak\nonumber\\
& \phantom{=====} + \int_{0}^{t}\!\int_{S_{d}^{+}\! \setminus \left\{ 0\right\}}\!\operatorname{Tr}\!\left[\sigma\!\left(s,T\right) \xi\right] \left(\mu^{X\!}-\nu^{\ast\!}\right)\!\left(ds,d\xi\right)\nonumber\\
& \phantom{=====} + 2 \int_{0}^{t}\! \operatorname{Tr}\left[\sigma\!\left(s,T\right) \sqrt{X_{s}}\, dW^{\ast}_{s}\, Q\right]\,,\nonumber
\end{align*}
\normalsize
\emph{where we have used again Proposition 1.28 of Chapter II in} \cite{Book_JacodShiryaev}.
\hfill
$\Box$
 \end{pf_appendix_2}


\vspace{0.5cm}
\bibliography{biblio}
\bibliographystyle{plain}

\end{document}